\documentclass{fundam}
\usepackage{url} 
\usepackage[ruled,lined]{algorithm2e}
\usepackage{graphicx}

\usepackage{braket}
\usepackage{amsmath}
\usepackage{amsfonts}
\usepackage{amssymb}
\usepackage{multirow}
\usepackage{booktabs}
\usepackage{float}

\begin{document}
\setcounter{page}{0001}
\issue{XXI~(2020)}

\title{Polynomial representation of general partial Boolean functions with a single quantum query}

\address{School of Computer Science and Engineering, Sun Yat-sen University, Guangzhou 510006, China\\}

\author{Xu Guoliang\\ 
 College of Information Technology, \\ Luoyang Normal University, Luoyang 471934, China; \\
   Institute of Quantum Computing  and Computer Theory, School of Computer Science and Engineering, \\
   Sun Yat-sen University, Guangzhou 510006, China.
   xu1guo2liang@foxmail.com
\and Qiu Daowen \thanks{ Corresponding author. This work is supported in part by the National Natural Science Foundation of China (Nos.61572532, 61876195, 62272208), the Natural Science Foundation of Guangdong Province of China
(No.2017B030311011), the Natural Science Foundation of Henan (No.232300420426), the Science and Technology Innovation Team of Henan University (No.22IRTSTHN016), the Key Scientific Research Project of Higher Education of Henan Province (Nos.23A520013, 22A110014), the Guangxi Key Laboratory of Trusted Software (Grant No.KX202040).}\\
Institute of Quantum Computing and Computer Theory,  School of  Computer Science and Engineering, \\ Sun Yat-sen University,  Guangzhou 510006, China \\
 Corresponding author address: issqdw@mail.sysu.edu.cn }\maketitle

\runninghead{G. Xu, et al.}{Polynomial representation of general partial Boolean functions with a single quantum query}

\begin{abstract}
Early in 1992,
Deutsch-Jozsa algorithm computed a symmetric partial Boolean function with a single quantum query, and thus achieved the best separation between classical deterministic and exact quantum query complexity.
Until recent years,
it was clarified that all
symmetric partial Boolean functions with a single quantum query can be computed exactly by Deutsch-Jozsa algorithm. For the general partial Boolean functions with a single quantum query, the latest characterizations is complex and not very satisfactory.
Based on this, this paper proves and discovers three new results:
(1) Establishing a new equivalence, each partial Boolean function with a single quantum query can be transformed to a simple partial Boolean function whose polynomial degree is just one;
(2) For partial Boolean functions up to four bits, there are only 10 non-trivial partial Boolean functions with a single quantum query;
(3) For each quantum 1-query algorithm with undefined measurement, there exists a constructive method for finding out all partial Boolean functions that can be computed exactly by the algorithm.
Essentially, the first discovery represent a step forward for a fundamental conclusion that the polynomial degree of partial Boolean functions with a single quantum query is one or two, and the last two results contribute a way for searching more nontrival partial Boolean functions that have quantum advantages.
\end{abstract}


\section{Introduction}
\label{Intro}

In quantum computation, quantum query model (i.e., quantum black box model or quantum decision tree model) \cite{Buhrman2002Complexity,Nielson2000Quantum} captures most of the known quantum algorithms, such as Deutsch's algorithm \cite{Deutsch1985quantumtheory}, Deutsch-Jozsa algorithm \cite{Deutsch1992rapidsolution}, Shor's factoring algorithm \cite{Shor1994Algorithms}, Grover's unstructured search algorithm \cite{Grover1996Afast}, HHL algorithm for linear systems of equations \cite{HHL2009algorithm} and many others \cite{quantumalgorithmzoo}. As the first quantum (query) algorithm, Deutsch's algorithm offers an essential theoretical framework for the construction of subsequent famous algorithms, and can be used as a subroutine for developing new large-scale quantum query algorithms \cite{exactdewolf2021}.
Certainly, in recent years, these algorithms still attract the attention of researchers \cite{Qiu2016revisit1,Qiuxu2021SPIN,quantumdistributed2021}.

Quantum query model computes a Boolean function  $f(x_{1},\cdots,x_{n})$, by means of accessing variables $x_{i}$ via quantum queries. Naturally, the complexity of an algorithm measures the number of quantum queries that it costs. For all legal inputs, an exact quantum query algorithm always can output correct function values \cite{Buhrman2002Complexity}, while a bounded-error quantum algorithm outputs function values with probabilities greater than $\frac{2}{3}$.
Based on this, the exact (or bounded-error) quantum query complexity $Q_{E}(f)$ (or $Q(f)$) of a Boolean function $f$ denotes the number of queries of an optimal quantum query algorithm that computes $f$ exactly (or with bounded-error) \cite{Buhrman2002Complexity}. Over the past decade, the equivalence relationships between the bounded-error quantum query model and polynomials had been characterized by Refs.~\cite{Aaronson2016ccc,Arunachalam2019siam}, the breadth and depth quantum advantage of exact quantum query algorithms had been determined by Refs.~\cite{Ambainis2015qic,ambainis2016siamjc}, and the exact quantum query complexity of some Boolean functions had been investigated by Refs.~\cite{exactdewolf2021,Montanaro2015On,Ambainis2017Exact,Qiu2018Generalized,He2018Exact,Qiu2016character,Grillo2018Quantum}, and so on. Although, the quantum query complexity is still a fundamental research direction, and needs more thorough and comprehensive research \cite{Aaronsonopen2021,Belovsquery2020,Liquery2019}.

The exact quantum algorithms are something not very well studied and worth studying.
For general partial Boolean functions, this paper works on the fundamental case: Exact quantum 1-query algorithms or exact quantum algorithm with a single query.
For completeness, in Section \ref{preli}, both definitions and relations of total, symmetric partial and partial Boolean functions are classified.
If a non-constant Boolean function $f$ can be computed exactly by a quantum 1-query algorithm, then we say that the exact quantum query complexity of $f$ is one, i.e., $Q_{E}(f)=1$.
There have been some work in this direction previously.
Initially, Deutsch's algorithm and Deutsch-Jozsa algorithm are exact quantum 1-query algorithms.
In 2007, Montanaro \cite{Montanaro2007thesis} considered a problem of exact oracle identification with a single quantum query.
In 2015, Montanaro et al. \cite{Montanaro2015On} clarified all small total Boolean functions up to four bits and all symmetric total Boolean functions up to six bits.
Being aware that Deutsch-Jozsa problem can be equivalently described as a symmetric partial Boolean function, Qiu et al. \cite{Qiu2016revisit1,Qiu2016character} in 2016 found all symmetric partial Boolean functions with $Q_{E}(f)=1$, and showed that these functions can be computed exactly by Deutsch-Jozsa algorithm \cite{Deutsch1992rapidsolution}.
In 2020, Mukherjee et al. \cite{Mukherjee2020Classical} pointed out that all total Boolean functions with $Q_{E}(f)=1$ were implied in \cite{Montanaro2015On}.
With a simple argument, Qiu et al. \cite{Qiuxu2021SPIN} in 2021 proved that all symmetric partial Boolean functions with $Q_{E}(f)=1$ include all total Boolean functions with $Q_{E}(f)=1$.
As a result, all symmetric partial Boolean functions and total Boolean functions with $Q_{E}(f)=1$ can be computed exactly by Deutsch-Jozsa algorithm.
However, for all partial Boolean functions with $Q_{E}(f)=1$, things are quite complicated.
In 2016, Grillo et al. \cite{Grillo2018Quantum} established a theoretical linkage between a system of equations and the exact quantum $t$-query algorithm.
In 2020, Xu et al. \cite{Xuqiu2020partial} provided two sufficient and necessary conditions to characterize any $n$-bit partial Boolean function with $Q_{E}(f)=1$,
and showed that the number of all $n$-bit partial Boolean functions with $Q_{E}(f)=1$ is upper bounded by a function in $n$.
Recently, Qiu et al. \cite{Qiuxu2021SPIN} constructed a new partial Boolean function with $Q_{E}(f)=1$ that cannot be computed exactly by Deutsch-Jozsa algorithm.
As a consequence, the form of all partial Boolean functions with $Q_{E}(f)=1$ is very interesting but still not clear.

The main observation of this paper is that the properties of the function $f$ depend on the set $S = \{a\oplus b:f(a)=0~\textrm{and}~f(b)=1\}$, where $\oplus$ is bit-wise XOR. Hence, we can study an equivalent function $g_{f}$, defined by $g_{f}(x) = 1$ for all $x\in S$ and $g_{f}(0)=0$. Both function are simultaneously either computable or non-computable by exact quantum 1-query algorithm. For $g_{f}$, it is very easy to describe the corresponding polynomial. In particular, $g_{f}$ is essentially a real linear function. As above, it simplifies the analysis of such functions dramatically.



\subsection{Our results}
This paper works on partial Boolean functions computed by exact quantum 1-query algorithms. Specifically, we achieve the following results.
\begin{itemize}
  \item[(1)] Establishing a new equivalence relation, we show that all partial Boolean functions with $Q_{E}(f)=1$ can be transformed to a simple form. In the simple form, the partial Boolean functions output 0 if and only if the input is $00\cdots 0$. The related results are presented in Corollaries \ref{Reductionlaw1} and \ref{Reductionlaw2}. Based on this, we obtain Theorem \ref{Reduceddecisiontheorem} and Corollary \ref{aresult} which can improve the performance of Theorem 1 in Ref.~\cite{Xuqiu2020partial}.
  \item[(2)] Up to the new equivalence, we show that all transformed partial Boolean functions with $Q_{E}(f)=1$ can be represented by degree-1 multilinear polynomials. This result is presented in Theorem \ref{acorollaryresult}. Essentially, the result represents a step forward for a fundamental fact in Ref. \cite{Buhrman2002Complexity} that the polynomial degree of all partial Boolean functions with $Q_{E}(f)=1$ is one or two.
  \item[(3)] As an application, all 3-bit and 4-bit transformed partial Boolean functions with $Q_{E}(f)=1$ are checked by the degree-1 multilinear polynomials of Theorem \ref{acorollaryresult} (see Tables \ref{Tab03} and \ref{Tab04}). The details are presented in Corollaries \ref{3bitpart1query} and \ref{4bitpart1query}. These results imply that there are only 10 new non-trivial transformed partial Boolean functions with $Q_{E}(f)=1$. In contrast, Montanaro et al. \cite{Montanaro2015On} investigated numerically all small total Boolean functions up to four bits and all symmetric total Boolean functions up to six bits.
  \item[(4)] Based on the representation, we introduce a construction method (i.e., Theorem \ref{Constructiontheorem}) for finding out all transformed partial Boolean functions computed by a given exact quantum 1-query algorithm. In particular, we find out all partial Boolean functions computed by Deutsch-Jozsa algorithm, and the result shows that these functions can be transformed to a symmetric partial Boolean function with $Q_{E}(f)=1$.
\end{itemize}

\subsection{Organisation}
The remainder of the paper is organized as follows. In Section \ref{preli}, we recall  basic notations and related results that we need in this paper. Then, in Section \ref{theoremXu} we transform all partial Boolean functions with $Q_{E}(f)=1$ to a simple form. After that, in Section \ref{quantum1querycomplexity}, we prove a representation theorem of transformed partial Boolean functions with $Q_{E}(f)=1$. In addition, we introduces a construction theorem in Section \ref{intereexam}. Finally, the discussion and conclusion are presented in Section \ref{Concl}. For the sake of brevity and readability, some results are put in Appendixes.

\section{Preliminaries}\label{preli}

In this section, we introduce related notations and recall  basic properties concerning partial Boolean functions and  exact quantum query model. For the details, we can refer to
Refs.\cite{Buhrman2002Complexity,Nielson2000Quantum,Qiuxu2021SPIN,Qiu2016character}.

\subsection{Boolean functions and polynomials}\label{boolandpol}

In this paper, we mainly concern $n$-bit partial Boolean functions  $f: D\rightarrow \{0,1\}$, where $D\subseteq\{0,1\}^{n}$ \cite{Qiuxu2021SPIN}.
In particular, if $D=\{0,1\}^{n}$, then $f$ is called a total Boolean function.
For any input $x=x_{1}x_{2}\cdots x_{n}\in D$, the Hamming weight (number of 1s) of $x$ is denoted by $\lvert x \rvert$. Symmetric partial Boolean functions \cite{Qiu2016revisit1,Qiu2018Generalized,Qiu2016character} are partial Boolean functions $f: D\rightarrow\{0,1\}$ satisfying the following two conditions: (1) $f(x)=f(y)$ for all $\lvert x\rvert=\lvert y\rvert$ where $x,y\in D$; (2) If $\lvert x\rvert=\lvert y\rvert$, then $x\in D$ if and only if $y\in D$ (Here, one may feel a bit confused why give such a strong definition for symmetric function. For example, Deutsch-Jozsa algorithm works even if the we are distinguishing $\{0^n\}$ v.s. a subset of $\{x : |x| = n/2\}$. In fact, the weaker symmetric function is not more difficult than Deutsch-Jozsa problem. In other word, for general symmetric partial Boolean functions, one may find out a better algorithm to solve a weaker symmetric function. However, for symmetric functions with a single quantum query, one can not find out a better algorithm unless the weaker symmetric function is a constant). Therefore, the relation of partial, symmetric partial and total Boolean functions can be figured in Fig. \ref{figure1}.
\begin{figure}[h]%
\centering
\includegraphics[width=0.7\textwidth]{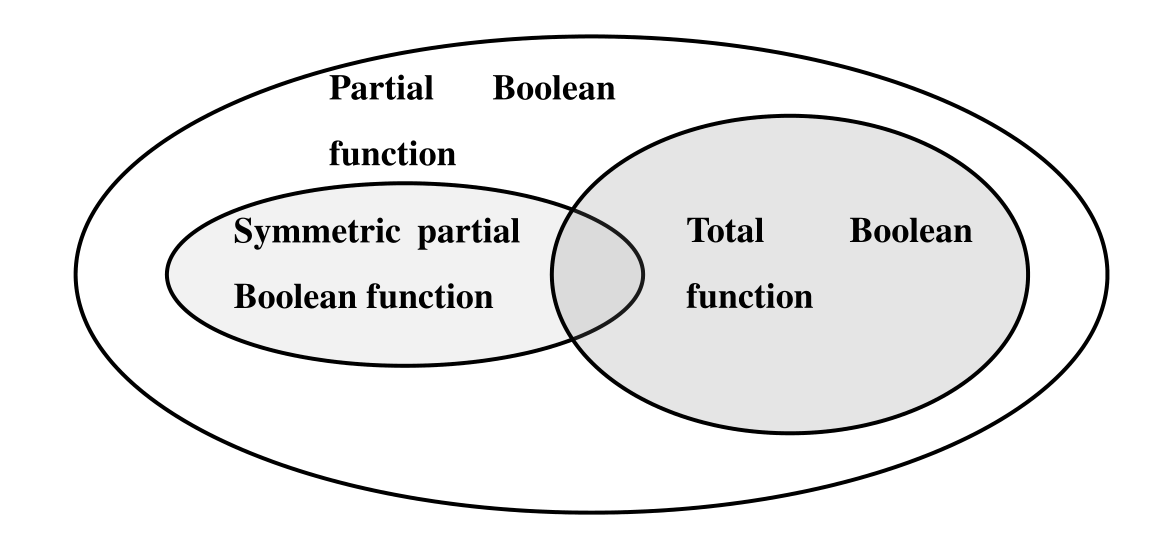}
\caption{The relation of partial, symmetric partial and total Boolean functions.}\label{figure1}
\end{figure}

In general, an $n$-bit partial Boolean function $f: D\rightarrow \{0,1\}$ can be given by a $2^{n}$-dimensional vector
$(f(0),f(1),f(2),\cdots,f(2^{n}-1))^{T}$
whose entry $f(x)$ is $*$ for any undefined input $x\in\{0,1\}^{n}\backslash D$.
For example, the Boolean function $f$ computed by Deutsch's algorithm \cite{Deutsch1985quantumtheory} can be given by $(f(00),f(01),$\\
$f(10),f(11))=(1,0,0,1)$.
Sometimes, we use a two-tuple
$(\{a:f(a)=0\},\{b:f(b)=1\})$
to give a partial Boolean function $f: D\rightarrow\{0,1\}$.
For example, the even $n$-bit partial Boolean function $f$ computed by Deutsch-Jozsa algorithm \cite{Deutsch1992rapidsolution} can be given by $(\{a:\lvert a \rvert=0,n\},\{b:\lvert b \rvert=\frac{n}{2}\})$.

To represent a partial Boolean function, we use two monomials $X_{S}=\prod_{i\in S}x_{i}$ and $(-1)^{S\cdot x}=\prod_{i\in S}(-1)^{x_{i}}$ where $X_{\emptyset}=(-1)^{\emptyset\cdot x}=1$ \cite{Buhrman2002Complexity}.
In general, the set $\{X_{S}:S\subseteq\{1,2,\cdots,n\}\}$ is  called the polynomial basis and the set $\{(-1)^{S\cdot x}:S\subseteq\{1,2,\cdots,n\}\}$ is called Fourier basis \cite{Buhrman2002Complexity}.
If a function $p:\mathbb{R}^{n}\rightarrow \mathbb{C}$ can be written as $\Sigma_{S}\alpha_{S}X_{S}$ for some complex numbers $\alpha_{S}$, then the function $p$ is called a multilinear polynomial \cite{Buhrman2002Complexity}.
For any partial Boolean function $f:D\rightarrow \{0,1\}$, a multilinear polynomial $p(x)$ represents $f$ if and only if $p(x)=f(x)$ for all $x\in D$ \cite{Buhrman2002Complexity,NisanOn}.
Unlike total functions $f:\{0,1\}^{n}\rightarrow  \{0,1\}$, the multilinear representation of a partial (not total) function $f:D\rightarrow \{0,1\}$ may be not unique \cite{Qiuxu2021SPIN}.

For any finite set $S$, the notation $\lvert S\rvert$ denotes the number of elements in $S$.
For a complex matrix $A$, $A^{T}$ is the transpose of the matrix $A$, and $A^{\dag}=(A^{*})^{T}$ is the conjugate transpose of the matrix $A$. Obviously, $A^{\dag}=A^{T}$ holds for any real matrix $A$.
In addition, the notation $\lvert a\rangle$ is usually used to denote a quantum state which is a unit vector and $\langle a \rvert=(\lvert a\rangle)^{\dag}$ is a row vector.

\subsection{Quantum query model and 1-query algorithm}\label{querymodelandalgori}

In quantum query model, for every input $x\in D$, the quantum black box $O_{x}$ can be described as a unitary operator which is defined by
\begin{equation}\label{actionofox}
O_{x}\lvert i,j\rangle=\left\{
\begin{array}{lcl}
(-1)^{x_{i}}\lvert i,j\rangle, & & \textrm{if}~~i\in\{1,2,\cdots,n\},\\
\lvert 0,j\rangle, & & \textrm{if}~~i=0.
\end{array} \right.
\end{equation}
Here, the integer number $i\in\{0,1,2,\cdots,n\}$ is the query-part and the label $j$ is the other-part.
Then, a quantum $t$-query algorithm is determined by an initial state $\lvert\psi_{0}\rangle$ and a sequence of unitary transformations $U_{0},$ $O_{x},$ $U_{1},$ $O_{x},$ $\cdots,$ $O_{x},$ $U_{t}$ followed by a measurement, where $t+1$ unitary operators $U_{0},U_{1},\cdots,U_{t}$ are independent of the input \cite{Buhrman2002Complexity,Nielson2000Quantum}.

This paper concerns the quantum 1-query algorithm which is determined by an initial state $\lvert\psi_{0}\rangle$ and three unitary transformations $U_{0},$ $O_{x},$ $U_{1}$ followed by a measurement. In Fig. \ref{figure2}, a circuit of quantum 1-query algorithms is described.
\begin{figure}[h]%
\centering
\includegraphics[width=0.9\textwidth]{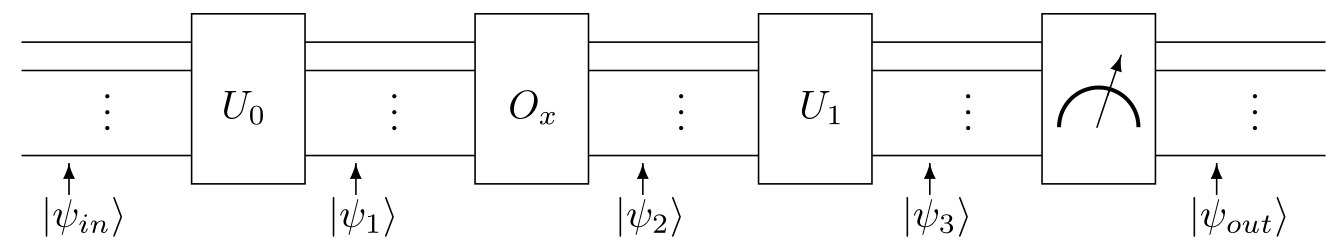}
\caption{A circuit of quantum 1-query algorithms.}\label{figure2}
\end{figure}

Assume that before making the quantum query, the algorithm prepares a state $|\psi\rangle$. If after a single quantum query, the 1 instances can be distinguished from the 0 instances perfectly, it has to be the case, that $O_x|\psi\rangle$ form two orthogonal subspaces spanned by the 1 instances and the 0 instances, where $O_x$ is the query oracle that puts phases based on the ``address" register. Open up everything, it means that for some unit vector $\vec{a} = 0$ for any 1 instance $x$ and 0 instance $y$. In other words, for some non-negative real vector $z$ with unit $l$-1 norm, such that $\vec{a}=0$. Using this fact, observe that for any such function (computable exactly by 1-query quantum algorithm), the set $S = \{a\oplus b:f(a)=0~\textrm{and}~f(b)=1\}$ determines an equivalence class. And the representative function for each equivalence class $S$, is such that $g_{f}(0)=0$, $g_{f}(x)=1$ for $x\in S$. Further, any representative function can be described exactly by a degree-1 polynomial alluded in the previous paragraph.

Based on these observations, to determine all such functions, it boils down to first determine all possible $S$ and the corresponding representative function that admits an exact 1-query algorithm. This is to decide whether there exists a non-negative real vector $z$ satisfying the linear constraints mentioned previously. The authors thus list all such $S$ for the case when $n=3,4$. (the case when $n=1,2$ is understood). The details are in the following.

\section{Equivalent transformation of partial Boolean functions}\label{theoremXu}
In this section, we establish a new equivalence relation on partial Boolean functions with $Q_{E}(f)=1$ and then give more characterizations of them.

First, it is necessary to recall Theorem 1 of Ref. \cite{Xuqiu2020partial}. For convenience to describe, we define the vector
\begin{equation}\label{}
\phi(a)=(1,(-1)^{a_{1}},(-1)^{a_{2}},\cdots,(-1)^{a_{n}})^{T}
\end{equation}
which corresponds uniquely to an input string $a=a_{1}a_{2}\cdots a_{n}\in\{0,1\}^{n}$.
\begin{theorem}[Decision theorem]\label{Decisiontheorem}
\cite{Xuqiu2020partial}. An $n$-bit non-constant partial Boolean function $f:$ $D\rightarrow$ $\{0,1\}$ is computed by an exact quantum 1-query  algorithm, if and only if there exist at least one non-negative vector $\mathbf{z}=$($z_{0},$$z_{1},$$z_{2},$$\cdots,$$z_{n})^{T}$ such that the linear system of equations
\begin{equation}\label{solve}
\begin{split}
\left \{
 \begin{aligned}
&\mathbf{z}^{T}\phi(0)=1, &\\
&\mathbf{z}^{T}\phi(a\oplus b)=0, & \forall a, b\in D~\textrm{and}~f(a)\neq f(b).
 \end{aligned}
\right.
\end{split}
\end{equation}
holds.
\end{theorem}
In fact, Theorem \ref{Decisiontheorem} establishes a relationship between the linear system of equations and an exact quantum 1-query algorithm by Eq.~\eqref{solve}. Here, $z_{i}=\sum_{j}\lvert\alpha_{i,j}\rvert^{2}$ holds for $U_{0}\lvert\psi\rangle=\sum_{i,j}\alpha_{i,j}\lvert i,j\rangle$ (i.e., the state before the unique query of an exact quantum 1-query algorithm).
In fact, if a quantum 1-query algorithm in this class and a Boolean function $f$ can make Eq.~\eqref{solve} hold, then it is not difficult to verify that the algorithm can compute $f$ exactly.

By Theorem \ref{Decisiontheorem}, we get the following two transformation laws that will contribute to define a new equivalence relation.
\begin{corollary}[Transformation law 1]\label{Reductionlaw1}
If an $n$-bit non-constant partial Boolean function $f:$ $D\rightarrow$ $\{0,1\}$ can be computed  by an exact quantum 1-query  algorithm, then for any $S\subseteq\{a\oplus b:f(a)\neq f(b)\}$, the function $g_{f}$ defined by
\begin{equation}\label{}
g_{f}(x)=\left\{
\begin{array}{lcl}
1, & & x\in S,\\
0, & & x=00...0
\end{array} \right.
\end{equation}
can also be computed  by the same exact quantum 1-query  algorithm.
\end{corollary}
\begin{proof}
By Theorem \ref{Decisiontheorem}, if $f:$ $D\rightarrow$ $\{0,1\}$ can be computed by an exact quantum 1-query algorithm, then the linear system of equations $z_{0}+$$z_{1}+$$z_{2}+$$\cdots+$$z_{n}$=1 and $\mathbf{z}^{T}\phi(a\oplus b)=0$ for all $a\in\{x:f(x)=0\}$ and $b\in\{x:f(x)=1\}$ has one non-negative solution. Meanwhile, the non-negative solution is also a solution of the linear system of equations $z_{0}+$$z_{1}+$$z_{2}+$$\cdots+$$z_{n}$=1 and $\mathbf{z}^{T}\phi(a\oplus b)=0$ for all $a\in\{x:g(x)=0\}$ and $b\in\{x:g(x)=1\}$. Thus, the corollary has been proved.
\end{proof}
For $S=\{a\oplus b:f(a)\neq f(b)\}$, the inverse direction of Corollary \ref{Reductionlaw1} also holds, since the two functions share a common linear system of equations.
\begin{corollary}[Transformation law 2]\label{Reductionlaw2}
An $n$-bit non-constant partial Boolean function $f:$ $D\rightarrow$ $\{0,1\}$ can be computed  by an exact quantum 1-query  algorithm, if and only if the function $g_{f}$ defined by
\begin{equation}\label{}
g_{f}(x)=\left\{
\begin{array}{lcl}
1, & & x\in \{a\oplus b:f(a)\neq f(b)\},\\
0, & & x=00...0
\end{array} \right.
\end{equation}
can be computed by an exact quantum 1-query  algorithm.
\end{corollary}
In Discrete mathematics, a Binary relation on a nonempty set is an equivalence relation, if the relation is reflexive, symmetric, and transitive.
In the following, we define a Binary relation on the set of all partial Boolean functions $f:D\rightarrow\{0,1\}$ where $D\subseteq\{0,1\}^{n}$ and the three properties (i.e., reflexivity, symmetry and transitivity) can be verified by Corollaries \ref{Reductionlaw1} and \ref{Reductionlaw2}.
\begin{definition}\label{def111000}
For two partial Boolean functions $f$ and $g$, if each of them can be computed by an exact quantum 1-query algorithm and the set $\{a\oplus b:f(a)\neq f(b)\}=\{a\oplus b:g(a)\neq g(b)\}$, then we say that $f$ is equivalent to $g$.
\end{definition}

\begin{remark}\label{reductremark}
By Corollary \ref{Reductionlaw2}, any partial Boolean function with $Q_{E}(f)=1$ can always be transformed to a partial Boolean function which outputs 0 if and only if the input is $00\cdots0$.
That is, for the exact quantum 1-query model, it is enough to investigate all partial Boolean functions of the form
\begin{equation}\label{simpleooooo}
f(x)=\left\{
\begin{array}{lcl}
1~\textrm{or}~*, & & x\neq 00...0,\\
0, & & x=00...0.
\end{array} \right.
\end{equation}
Thus, we always discuss the set $\{a:f(a)=1\}$ for a partial Boolean function with $Q_{E}(f)=1$.
\end{remark}

Applying Theorem \ref{Decisiontheorem} to Eq.~\eqref{simpleooooo}, the following transformed version of Theorem \ref{Decisiontheorem} is obtained. In addition, a related result is put into \ref{arelatedresulttheorem}.
\begin{theorem}[Transformed decision theorem]\label{Reduceddecisiontheorem} Let $f$ be an $n$-bit non-constant partial Boolean function whose value is 0 if and only if $x=00...0$. Then, $f$ can be computed  by an exact quantum 1-query  algorithm, if and only if there exist at least one non-negative solution $\mathbf{z}=(z_{1},z_{2},\cdots,z_{n})^{T}$ of the linear system of equations
\begin{equation}\label{}
\begin{split}
\left \{
 \begin{aligned}
&z_{1}+z_{2}+\cdots+z_{n}\leq1, &\\
&\sum_{i:x_{i}=1}z_{i}=\frac{1}{2}, & \forall x\in\{a:f(a)=1\}.
 \end{aligned}
\right.
\end{split}
\end{equation}
\end{theorem}
\begin{proof}
Replacing the set $\{a:f(a)=0\}$ in Theorem \ref{Decisiontheorem} with $\{00\cdots0\}$, we get the linear system of equations $z_{0}+z_{1}+z_{2}+\cdots+z_{n}=1$ and $\mathbf{z}^{T}\phi(b)=0$ for all $b\in\{x:f(x)=1\}$.
Subtracting $\mathbf{z}^{T}\phi(b)=z_{0}+\sum_{i:b_{i}=0}z_{i}-\sum_{i:b_{i}=1}z_{i}$
from $z_{0}+z_{1}+\cdots+z_{n}=z_{0}+\sum_{i:b_{i}=0}z_{i}+\sum_{i:b_{i}=1}z_{i},$ the proof has been finished.
\end{proof}

Finally, combining above results with previous works \cite{Deutsch1992rapidsolution,Qiu2016revisit1,Qiuxu2021SPIN,Qiu2016character,Xuqiu2020partial}, the following $n$-bit partial Boolean functions are known.
\begin{itemize}
\item [(1)] For an even $n$, $f(x)=1$ if and only if $\lvert x\rvert=\frac{n}{2}$ \cite{Deutsch1992rapidsolution}. This is the well-known Deutsch-Jozsa problem.
\item [(2)] Given an integer $c\in\{\lceil\frac{n}{2}\rceil,\cdots,n\}$, $f(x)=1$ if and only if $\lvert x\rvert=c$ \cite{Qiu2016revisit1,Qiu2016character}. Padding $2c-n$ zeros to all legal inputs, $f$ can be transformed to Deutsch-Jozsa problem.
\end{itemize}
In other word, we only known a few partial Boolean functions with $Q_{E}(f)=1$. However, an upper bound on the number of $n$-bit partial Boolean functions with $Q_{E}(f)=1$ is big \cite{Xuqiu2020partial}. These results motivate us to find some new non-trivial examples which meet the following two conditions.
\begin{itemize}
\item [(1)] The exact quantum 1-query  algorithm of computing the function depends on $n$ bits, i.e., $z_{1}z_{2}\cdots z_{n}>0$. In fact, if there exists a solution such that some of $z_{1},z_{2},\cdots,z_{n}$ are zero, then corresponding bits can be removed and the function is identified with a smaller Boolean function depending on all bits.
\item [(2)] The partial Boolean function is not a symmetric partial Boolean function and can not be computed by Deutsch-Jozsa algorithm.
\end{itemize}

\section{Representation of partial Boolean functions}\label{quantum1querycomplexity}
As in Corollary \ref{Reductionlaw2}, any partial Boolean function with $Q_{E}(f)=1$ can always be transformed to a partial Boolean function which outputs 0 if and only if the input is $00\cdots0$. In this section we introduce the following result which is applicable to the transformed partial Boolean functions with $Q_{E}(f)=1$.
\begin{theorem}[Representation theorem]\label{acorollaryresult}
For a partial Boolean function $f$ in the form of Eq.~\eqref{simpleooooo}, $f$ can be computed  by an exact quantum 1-query  algorithm if and only if $f$ can be represented by the multilinear polynomial
\begin{equation}\label{fcanbepresen00000}
p(x)=2a_{1}x_{1}+2a_{2}x_{2}+\cdots +2a_{n}x_{n}
\end{equation}
for some non-negative real numbers $a_{1}, a_{2}, \cdots, a_{n}$ satisfying $\sum^{n}_{i=1}a_{i}\leq 1$.
\end{theorem}
\begin{proof}
$\Rightarrow)$. By Theorem \ref{Reduceddecisiontheorem}, if $f$ can be computed by an exact quantum 1-query algorithm, then there exists at least one non-negative solution $\mathbf{z}=(z_{1},z_{2},\cdots,z_{n})^{T}$ of the linear system of equations
\begin{equation}\label{11111111111111}
\begin{split}
\left \{
 \begin{aligned}
&z_{1}+z_{2}+\cdots+z_{n}\leq1, &\\
&\sum_{i}z_{i}x_{i}=\frac{1}{2}, & \forall x\in\{b:f(b)=1\}.
 \end{aligned}
\right.
\end{split}
\end{equation}
For any solution $(z_{1},z_{2},\cdots,z_{n})^{T}$, without loss of generality, let $a_{1}=z_{1}, a_{2}=z_{2}, \cdots, a_{n}=z_{n}$.
By Eq.~\eqref{11111111111111},
\begin{equation}\label{fcanbepresent}
\begin{split}
p(x)=\sum^{n}_{i=1}2a_{i}x_{i}=2\sum_{i}z_{i}x_{i}=1=f(x)
\end{split}
\end{equation}
holds for all $x\in\{b:f(b)=1\}$. Meanwhile, $p(x)=0$ for $x=00\cdots 0$. Thus, $p(x)$ represents $f$.

$\Leftarrow)$. If a non-constant partial Boolean function $f$ can be represented by Eq.~\eqref{fcanbepresen00000}
for some non-negative real numbers $a_{1}, a_{2}, \cdots, a_{p}$ satisfying $\sum^{n}_{i=1}a_{i}\leq 1$, then
\begin{equation}\label{functiondeutschJ000000000}
f(x)=\left\{
\begin{array}{lcl}
0~\textrm{or}~*, & & \sum^{n}_{i=1}2a_{i}x_{i}=0,\\
1~\textrm{or}~*, & & \sum^{n}_{i=1}2a_{i}x_{i}=1,\\
*, & & \sum^{n}_{i=1}2a_{i}x_{i}\not\in\{0,1\}.
\end{array} \right.
\end{equation}
Since $a_{i}\geq 0$, we have $\sum^{n}_{i=1}2a_{i}x_{i}=0$ if and only if $x=00\cdots0$. Considering that $f$ is non-constant,
\begin{equation}\label{functiondeutschJ0000000}
f(x)=\left\{
\begin{array}{lcl}
0, & & x=00\cdots0,\\
1~\textrm{or}~*, & & \sum^{n}_{i=1}a_{i}x_{i}=\frac{1}{2},\\
*, & & \sum^{n}_{i=1}2a_{i}x_{i}\not\in\{0,1\}.
\end{array} \right.
\end{equation}
Take $z_{1}=a_{1},$ $z_{2}=a_{2},$ $\cdots,$ $z_{n}=a_{n}$, and the proof is finished by Theorem \ref{Reduceddecisiontheorem}.
\end{proof}
\begin{remark}\label{degreemultilinear}
Take the (total) 2-bit parity function (i.e., Deutsch's problem)
\begin{equation}\label{}
(f(00),f(01),f(10),f(11))=(0,1,1,0)
\end{equation}
as an example. As we know that there does not exist a degree-1 multilinear polynomial representing $f$. Now, we show that $f$ is equivalent to a partial Boolean function with degree one. First, since $Q_{E}(f)=1$ and $\{a\oplus b: f(a)\neq f(b)\}=\{01,10\}$, $f$ is equivalent to
\begin{equation}\label{}
g_{f}(x)=\left\{
\begin{array}{lcl}
1, & & x\in \{01,10\},\\
0, & & x=00,\\
*, & & x=11
\end{array} \right.
\end{equation}
by Definition \ref{def111000}. Obviously, the equation $p(x)=x_{1}+x_{2}=g_{f}(x)$ holds for all $x\in\{00,01,10\}$. Thus, $p$ represents $g_{f}$ and the polynomial degree of $g_{f}$ is one.
As a result, $f$ is equivalent to the partial Boolean function $g_{f}$ with degree one.
Furthermore, for $p(x)$, we can see that the coefficients $2a_{1}=2a_{2}=1$ (in Theorem \ref{acorollaryresult}) implies $a_{1}+a_{2}=1\leq 1$.

In the inverse direction of Theorem \ref{acorollaryresult}, the condition that $a_{1}, a_{2}, \cdots, a_{n}$ are non-negative real numbers is necessary. Take the three-bit partial Boolean function
\begin{equation}\label{}
f(x)=\left\{
\begin{array}{lcl}
1, & & x\in\{001,010,111\},\\
0, & & x=000,\\
*, & & \textrm{else}
\end{array} \right.
\end{equation}
as an example. It is not difficult to know that $f$ has a unique degree 1 multilinear polynomial representation
$p(x)=-x_{1}+x_{2}+x_{3}.$
Here, the coefficient -1 is less than 0. Meanwhile, the exact quantum query complexity of $f$ is bigger than 1 by Theorem \ref{Reduceddecisiontheorem}.
\end{remark}

Finally, we check transformed 3-bit and 4-bit partial Boolean functions with $Q_{E}(f)=1$ in \ref{Applicationofacorollaryresult}, and the result is enumerated in Tables \ref{Tab03} and \ref{Tab04}.

\section{Construction of partial Boolean functions}\label{intereexam}

This section presents a construction for finding out new non-trivial examples.
We start by giving the following definition.
\begin{definition}\label{def000}
A quantum 1-query algorithm is called a $(z_{1},z_{2},\cdots,z_{n})$ quantum 1-query algorithm, if the state $U_{0}\lvert\psi\rangle=\sum_{i,j}\alpha_{i,j}\lvert i,j\rangle$ satisfies $z_{i}=\sum_{j}\lvert\alpha_{i,j}\rvert^{2}$ for all $i\in\{1,2,\cdots,n\}$.
\end{definition}

In Definition \ref{def000},
\begin{equation}\label{}
1\geq\sum^{n}_{i=1}z_{i}\geq \sum^{n}_{i:x_{i}=1}z_{i}=\frac{1}{2}
\end{equation}
can be got by Theorem \ref{Reduceddecisiontheorem}. Without loss of generality, assume that $z_{1}=z_{2}=\cdots=z_{k_{1}}=a_{1}, z_{k_{1}+1}=z_{k_{1}+2}=\cdots=z_{k_{2}}=a_{2}, \cdots, z_{k_{p-1}+1}=z_{k_{p-1}+2}=\cdots=z_{k_{p}}=a_{p}$. Now, we introduce the construction theorem of partial Boolean functions with $Q_{E}(f)=1$.
\begin{theorem}[Construction theorem]\label{Constructiontheorem}
Given $a_{1}, a_{2}, \cdots, a_{p}\in \mathbb{R}_{+}$ 
and $0=k_{0}<k_{1}<k_{2}<\cdots<k_{p}=n\in \mathbb{N}_{+}$ with $a_{1}k_{1}+a_{2}(k_{2}-k_{1})+\cdots +a_{p}(k_{p}-k_{p-1})\in[\frac{1}{2},1]$, then the function
\begin{equation}\label{functionconstruct}
f(x)=\left\{
\begin{array}{lcl}
1, & & x\in T,\\
0, & & x=00...0,\\
*, & & \textrm{else}
\end{array} \right.
\end{equation}
is an $n$-bit partial Boolean function with $Q_{E}(f)=1$ where the set $T$ is
\begin{equation}\label{allpart}
\{x:a_{1}\lvert x_{1}x_{2}\cdots x_{k_{1}}\rvert+\cdots+a_{p}\lvert x_{k_{p-1}+1}x_{k_{p-1}+2}\cdots x_{k_{p}}\rvert=\frac{1}{2}\}.
\end{equation}
\end{theorem}
\begin{proof}
Let $z_{1}=z_{2}=\cdots=z_{k_{1}}=a_{1},$ $z_{k_{1}+1}=z_{k_{1}+2}=\cdots=z_{k_{2}}=a_{2},$ $\cdots,$ $z_{k_{p-1}+1}=z_{k_{p-1}+2}=\cdots=z_{k_{p}}=a_{p}.$
Then, the linear system in Theorem \ref{Reduceddecisiontheorem} becomes
\begin{equation}\label{linearredu}
\begin{split}
\left \{
 \begin{aligned}
&z_{1}+z_{2}+\cdots+z_{n}=\sum^{p}_{i=1}a_{i}(k_{i}-k_{i-1})\leq1, &\\
&a_{1}\lvert x_{1}x_{2}\cdots x_{k_{1}}\rvert+\cdots+a_{p}\lvert x_{k_{p-1}+1}x_{k_{p-1}+2}\cdots x_{k_{p}}\rvert=\frac{1}{2}. &
 \end{aligned}
\right.
\end{split}
\end{equation}
Also,  Eq.~\eqref{linearredu} has the same solutions as
\begin{equation}\label{oureq}
\begin{split}
\left \{
 \begin{aligned}
&a_{1}m_{1}+a_{2}m_{2}+\cdots +a_{p}m_{p}=\frac{1}{2}, &\\
&m_{i}\in\{0,1,\cdots,k_{i}-k_{i-1}\},\forall i\in\{1,2,\cdots,p\} &\\
 \end{aligned}
\right.
\end{split}
\end{equation}
where $m_{1}=\lvert x_{1}x_{2}\cdots x_{k_{1}}\rvert$, $\cdots,$ $m_{p}=\lvert x_{k_{p-1}+1}x_{k_{p-1}+2}\cdots x_{k_{p}}\rvert$. Since every $(z_{1},z_{2},\cdots,z_{n})$ quantum 1-query algorithm computes exactly a partial Boolean function if and only if Eq.~\eqref{linearredu} holds, Eq.~\eqref{allpart} traverses all possible inputs of the partial Boolean function. Thus, the proof is completed.
\end{proof}
Note that the solution space of Eq.~\eqref{oureq} is $\{0,1,\cdots,k_{1}\}\times \{0,1,\cdots,k_{2}\}\times\cdots\times\{0,1,\cdots,k_{p}\}$ whose size is $(k_{1}+1)(k_{2}-k_{1}+1)\cdots(k_{p}-k_{p-1}+1)$. If $p$ is not big (For Deutsch-Jozsa algorithm, $p=1$, $a_{1}=\frac{1}{n}$ and $k_{1}=n$), Theorem \ref{Constructiontheorem} is quite practical. For $p=n$, the size of the solution space of Eq.~\eqref{oureq} is $2^{n}$ and Theorem \ref{Constructiontheorem} degenerates to Theorem \ref{Reduceddecisiontheorem}.

As above, given a quantum 1-query algorithm, Theorem \ref{Constructiontheorem} answers the following two questions.
\begin{itemize}
\item [(1)]~Can it be used to compute some partial Boolean functions exactly? This can be done by checking Eq.~\eqref{oureq}.
\item [(2)]~If so, what partial Boolean functions can be computed exactly by this quantum 1-query algorithm? All partial Boolean functions can be found out by the set in Eq.~\eqref{allpart}.
\end{itemize}
Thus, Theorem \ref{Constructiontheorem} characterizes the computational power of a given exact quantum 1-query algorithm for partial Boolean functions.

Finally, two applications of Theorem \ref{Constructiontheorem} are put into \ref{Applicationoftheconstructionmethod}.

%

\section{Discussion and conclusion}\label{Concl}

This paper studies the interesting problem of characterizing general partial functions with exact quantum query complexity equal to one. Previously, it was shown that symmetric partial Boolean functions and total Boolean functions with $Q_E(f)=1$ can be exactly computed by the Deutsch-Josza algorithm \cite{Qiu2016revisit1,Qiu2016character}, while for the case of general partial Boolean functions such a concrete result is not known. We, in our previous work \cite{Xuqiu2020partial}, characterized the partial functions with quantum query complexity equal to one using a system of linear equations. The present work extends the results provided in Ref. \cite{Xuqiu2020partial} by a simpler class of partial functions.
In particular, this paper proves that every partial function which can be computed by an exact quantum 1-query algorithm can be converted to a simple partial function which is zero only at the all zero input. The paper provides the characterization of these simple partial function in terms of a simplified system of linear equations (as compared with [Xu et. al, 2021]). It is then proved that the simple partial Boolean functions have a degree one multilinear polynomial representation, improving the upper bound of degree $\leq 2Q_E(f)=2$.
Using the aforementioned polynomial representation, this paper then provides a way to construct partial functions with $Q_E(f) = 1$ starting with a sequence of numbers having specified properties. The major contributions of this paper is that it simplifies the aforementioned problem with the help of theorems introduced in Ref. \cite{Xuqiu2020partial}.
Furthermore, the proofs that have been provided are mathematically correct and are easy to follow.
Specifically, the following two observations are quite attractive.

On one hand, the new and important equivalence can be used to investigate all partial Boolean functions with $Q_{E}(f)=1$. Undoubtedly, when we search for a new non-trivial partial Boolean function with $Q_{E}(f)=1$, the new equivalence helps us reduce the size of the search space.
For example, since a 4-bit partial Boolean function corresponds to a vector $(f(0),f(1),\cdots,f(2^{4}-1))\in\{0,1,*\}^{2^{4}}$ which has $2^{4}$ undetermined entries, the size of the search space reaches $3^{16}=43046721$
when searching for a new non-trivial 4-bit partial Boolean function with $Q_{E}(f)=1$.
By means of the new equivalence, any transformed 4-bit partial Boolean function corresponds to a vector $(f(0),f(1),\cdots,f(2^{4}-1))\in0\{1,*\}^{2^{4}-1}$ which has $2^{4}-1$ undetermined entries and the size $3^{2^{4}}$ degenerates to $2^{2^{4}-1}=2^{15}=32768$. It is clear that the size of the search space is reduced at least exponentially.

On the other hand, the new and simple representation is attractive and important from the following three aspects.
\begin{itemize}
\item [(1)]~The representation discovers a new notion of the multilinear polynomial degree that equals the exact quantum query complexity. In contrast, Arunachalam et al. \cite{Arunachalam2019siam} in 2017 obtained a new notion of approximate polynomial degree that equals the bounded-error quantum query complexity.
\item [(2)]~The representation breaks through the best result that the polynomial degree of all partial Boolean functions with $Q_{E}(f)=1$ is one or two. Indeed, the best result can be obtained by Theorem 17 of Ref.\cite{Buhrman2002Complexity}(i.e., the polynomial degree of $f$ is not bigger than $2Q_{E}(f)$). Although the number of all transformed 4-bit partial Boolean functions is still large, the task of checking all reduced $3$-bit and $4$-bit partial Boolean functions is completed well in a few pages.
\item [(3)]~The representation demonstrates a new equivalence relation between exact quantum algorithms and polynomials: the existence of an exact quantum 1-query algorithm computing a transformed partial Boolean function $f$ is equivalent to the existence of a degree-1 polynomial $p$ that represents $f$. Correspondingly, de Wolf \cite{WolfNondeterministic2003} in 2003 proved the equivalence relationship between nondeterministic quantum algorithms and nondeterministic polynomials, Montanaro et al. \cite{Montanaro2011Unbounded} in 2011 proved the equivalence relationship between unbounded-error quantum algorithms and threshold polynomials, and Aaronson et al. \cite{Aaronson2016ccc} in 2016 proved the equivalence relationship between bounded-error quantum 1-query algorithms and representations by degree-2 polynomials.
\end{itemize}

Finally, we have proposed a construction method which can be used to find out some new non-trivial examples and discover the computational power of Deutsch-Jozsa algorithm. In some ways, the construction method paves a way for finding out more Boolean functions that have quantum advantages. In contrast, the existing Boolean functions that have quantum advantages were not proposed constructively and thus difficult to be employed.

As above, the generalization of these discoveries can be considered and are worthy of further exploration. That is, we can ask the following two interesting questions.
\begin{itemize}
\item [(1)]~Are there some equivalence relations of partial Boolean functions with $Q_{E}(f)=2, 3, \cdots$? If so, how do we convert them?
\item [(2)]~Can we represent partial Boolean functions with $Q_{E}(f)=2, 3, \cdots$ by some polynomials and find out many non-trivial examples?
\end{itemize}

\section*{Acknowledgements}

This work is supported in part by the National Natural Science Foundation of China (Nos.61572532, 61876195, 62272208), the Natural Science Foundation of Guangdong Province of China
(No.2017B030311011), the Natural Science Foundation of Henan (No.232300420426), the Science and Technology Innovation Team of Henan University (No.22IRTSTHN016), the Key Scientific Research Project of Higher Education of Henan Province (Nos.23A520013, 22A110014), the Guangxi Key Laboratory of Trusted Software (Grant No.KX202040).

\appendix

\section{A corollary of Theorem \ref{Reduceddecisiontheorem}}\label{arelatedresulttheorem}

Combining Theorem \ref{Reduceddecisiontheorem} with Lemma 4 of Ref.~\cite{Xuqiu2020partial}, we get the following result.
\begin{corollary}\label{aresult}
For any $n$-bit partial Boolean function $f$ with $Q_{E}(f)=1$,
\begin{equation}\label{}
\lvert\{a\oplus b:f(a)\neq f(b)\}\rvert\leq 2^{n-1}.
\end{equation}
\end{corollary}
\begin{proof}
Let $S=\{a:f(a)=1\}\subseteq\{0,1\}^{n}$. By Theorem \ref{Decisiontheorem} there exist at least one non-negative solution of the equations
\begin{equation}\label{}
\begin{split}
\left \{
 \begin{aligned}
&\mathbf{1}^{T}\mathbf{z}=1, &\\
&\mathbf{z}^{T}\phi(a)=0, &a\in S.
 \end{aligned}
\right.
\end{split}
\end{equation}
Introducing the vector
$\varphi(a)=(1,a_{1},a_{2},\cdots,a_{n})^{T}$
for a string $a=a_{1}a_{2}\cdots a_{n}\in\{0,1\}^{n}$(Similar to $\phi(a)$, there also exists a one-to-one correspondence of the function vector $\varphi(a)$ and the input $a$), we have
\begin{equation}\label{}
\begin{split}
\left \{
 \begin{aligned}
&\mathbf{1}^{T}\mathbf{z}=1, &\\
&\mathbf{z}^{T}\varphi(a)=0, &a\in S.
 \end{aligned}
\right.
\end{split}
\end{equation}
Here, we used the fact that the polynomial coefficient vector can be transformed to the Fourier coefficient vector using an invertible matrix \cite{Xuqiu2021from}.
Next, there  exist at least one non-zero solution of the linear system of equations
\begin{equation}\label{}
\mathbf{z}^{T}\varphi(a)=0, a\in S.
\end{equation}
Using the solution theory of linear system of equations, we have
\begin{equation}\label{rankpx}
\textrm{rank}(\{\varphi(a):a\in S\})\leq n.
\end{equation}
Applying Lemma 4 of Ref. \cite{Xuqiu2020partial} (The Lemma says, ``Let the vector function $\varphi(X_{k})=(1,X_{k,1},X_{k,2},$\\
$\cdots,X_{k,n})^{T}$ for a string $X_{k}=X_{k,1}X_{k,2}\cdots X_{k,n}\in\{0,1\}^{n}$. If $n\geq 2$, for any $j$$(\in\{1,2,\cdots,n+1\})$ different basis vectors $\varphi(X_{1})$, $\varphi(X_{2})$, $\cdots,$ $\varphi(X_{j})$, there exist at most $T_{j}\leq2^{j-1}-j$ other different vectors $\varphi(X_{j+1})$, $\varphi(X_{j+2})$, $\cdots,$ $\varphi(X_{j+T_{j}})$ satisfying $\textrm{rank}([\varphi(X_{1}),\varphi(X_{2}),\cdots,\varphi(X_{j+T_{j}})])=j$.") to Eq.~\eqref{rankpx}, the proof is completed.
\end{proof}
\begin{remark}\label{exampledj}
By Corollary \ref{aresult}, if an $n$-bit partial Boolean function $f$ satisfies $\lvert\{a\oplus b:f(a)\neq f(b)\}\rvert>2^{n-1}$, then the exact quantum query complexity is bigger than 1. For an $n$-bit partial Boolean function $f$ satisfying $\lvert\{a\oplus b:f(a)\neq f(b)\}\rvert\leq2^{n-1}$, the number of equations in Theorem \ref{Reduceddecisiontheorem} is less than $2^{n-1}+1$. In contrast, the immediate number of equations in Theorem \ref{Decisiontheorem} is less than $2^{n}$. Thus, Corollary \ref{aresult} improves the efficiency of Theorem \ref{Decisiontheorem}.

Let us take the even $n$-bit partial Boolean function $(\{a:\lvert a\rvert=\frac{n}{2}\},\{b:\lvert b\rvert=0,n\})$ (i.e. Deutsch-Jozsa problem) as an example. Corollary \ref{aresult} works well on this function, since
\begin{equation}\label{}
\binom{n}{\frac{n}{2}}{=}\binom{n-1}{\frac{n}{2}}{+}\binom{n-1}{\frac{n}{2}-1}\leq \sum^{n-1}_{i=0}\binom{n-1}{i}{=}2^{n-1}.
\end{equation}
\end{remark}

\section{Applications of Theorem \ref{acorollaryresult}}\label{Applicationofacorollaryresult}
In this section, we employ Theorem \ref{acorollaryresult} to check 3-bit and 4-bit partial Boolean functions.

\subsection{Checking all 3-bit partial Boolean functions}\label{threebits}
In this subsection we check all 3-bit partial Boolean functions by Theorem \ref{acorollaryresult} and get the following result.
\begin{corollary}\label{3bitpart1query}
There do not exist  non-trivial 3-bit partial Boolean functions with $Q_{E}(f)=1$.
\end{corollary}
\begin{proof}
According to Theorem \ref{acorollaryresult}, a 3-bit partial Boolean function $f$ with $Q_{E}(f)=1$ is represented by a multilinear polynomial $p(x)=2a_{1}x_{1}+2a_{2}x_{2}+ 2a_{3}x_{3}$ for some non-negative real numbers $a_{1}, a_{2}, a_{3}$ satisfying $\sum^{3}_{i=1}a_{i}\leq 1$. Without loss of generality, assume that $a_{1}, a_{2}, a_{3}>0$. Thus,
\begin{equation}\label{p100p010}
\begin{split}
&(p(100),p(010),p(001),p(110),p(101),p(011),p(111))\\
=&(2a_{1},2a_{2},2a_{3},2(a_{1}+a_{2}),2(a_{1}+a_{3}),2(a_{2}+a_{3}),2(a_{1}+a_{2}+a_{3}))\in\{1,*\}^{7}.
\end{split}
\end{equation}
Now, we check all 3-bit multilinear polynomials in Eq.~\eqref{p100p010}. For convenience, all cases are listed in Table \ref{Tab03}.

\begin{table}
\begin{center}
\caption{All transformed 3-bit partial Boolean functions with $Q_{E}(f)=1$}\label{Tab03}%
\begin{tabular}{@{}lll@{}}
\toprule
Case & Boolean function: $\{b:f(b)=1\}$  & Remark \\
\midrule
1           &$\{100,011\}$                        & $f$ depends on two bits. \\

2           &                         & Similar to Case 1. \\

3           &$\{110,101,011\}$                          &   $f$ is symmetric. \\

4           &$\{101,011\}$                          &   $f$ depends on one bit.          \\

5           &$\{011\}$                             &   Included by Case 1.   \\

6           &$\{111\}$                           &    $f$ depends on one bit.  \\
\bottomrule
\end{tabular}
\end{center}
\end{table}

{\bf Case 1:} If $2a_{1}=1$, then $a_{1}, a_{2}, a_{3}>0$ and $\sum^{3}_{i=1}a_{i}\leq 1$ force that $2a_{2},2a_{3},2(a_{1}+a_{2}),2(a_{1}+a_{3}),2(a_{1}+a_{2}+a_{3}\neq 1$. Thus, the undetermined 3-bit partial Boolean function $f$ represented by Eq.~\eqref{p100p010} corresponds to
$(1,*,*,*,*,2(a_{2}+a_{3}),*).$
Set $2(a_{2}+a_{3})=1$ (For example, $a_{1}=\frac{1}{2},a_{2}=a_{3}=\frac{1}{4}$), then $p(100)=p(011)=1$ implies that $f$ depends on two bits (i.e., the first and the second bits) and can be computed exactly by Deutsch's algorithm.

{\bf Case 2:} By a similar argument, the cases that $2a_{1}\neq1, 2a_{2}=1$ and $2a_{1},2a_{2}\neq1,2a_{3}=1$ are similar to Case 1.

{\bf Case 3:} If $2a_{1}, 2a_{2}, 2a_{3}\neq1, 2(a_{1}+a_{2})=1$ and $a_{1}, a_{2}, a_{3}>0$ satisfying $\sum^{3}_{i=1}a_{i}\leq 1$, then the undetermined 3-bit partial Boolean function $f$ represented by Eq.~\eqref{p100p010} corresponds to
$(*,*,*,1,2(a_{1}+a_{3}),2(a_{2}+a_{3}),*).$
Set $2(a_{1}+a_{3})=2(a_{2}+a_{3})=1$ (i.e., $a_{1}=a_{2}=a_{3}=\frac{1}{4}$), then $f(110)=f(101)=f(011)=1$ implies that $f$ is a symmetric partial Boolean function which can be solved by Refs. \cite{Qiu2016character,Qiu2016revisit1}.

{\bf Case 4:} If $2a_{1}\neq1, 2a_{2}\neq1, 2a_{3}\neq1, 2(a_{1}+a_{2})\neq1, 2(a_{1}+a_{3})=1$ and $a_{1}, a_{2}, a_{3}>0$ satisfying $\sum^{3}_{i=1}a_{i}\leq 1$, then the undetermined 3-bit partial Boolean function $f$ represented by Eq.~\eqref{p100p010} corresponds to
$(*,*,*,*,1,2(a_{2}+a_{3}),*).$
Set $2(a_{2}+a_{3})=1$ (For example, $a_{1}=a_{2}=\frac{1}{6}$ and $a_{3}=\frac{1}{3}$), then $f(101)=f(011)=1$ imply that $f$ depends on the third bit and degenerates to a 1-bit Boolean function. 

{\bf Case 5:} If $2a_{1}\neq1, 2a_{2}\neq1, 2a_{3}\neq1, 2(a_{1}+a_{2})\neq1, 2(a_{1}+a_{3})\neq1, 2(a_{2}+a_{3})=1$ and $a_{1}, a_{2}, a_{3}>0$ satisfying $\sum^{3}_{i=1}a_{i}\leq 1$ (For example, $a_{1}=\frac{1}{8}$ and $a_{2}=a_{3}=\frac{1}{4}$), then the undetermined 3-bit partial Boolean function $f$ represented by Eq.~\eqref{p100p010} corresponds to
$(*,*,*,*,*,1,*).$
Thus, $f(011)=1$ implies that $f$ depends on the second (or the third) bit and degenerates to a 1-bit Boolean function.

{\bf Case 6:} If $2a_{1}\neq1, 2a_{2}\neq1, 2a_{3}\neq1, 2(a_{1}+a_{2})\neq1, 2(a_{1}+a_{3})\neq1, 2(a_{2}+a_{3})\neq1, 2(a_{1}+a_{2}+a_{3})=1$ and $a_{1}, a_{2}, a_{3}>0$ satisfying $\sum^{3}_{i=1}a_{i}\leq 1$ (For example, $a_{1}=a_{2}=a_{3}=\frac{1}{6}$), then the undetermined 3-bit partial Boolean function $f$ represented by Eq.~\eqref{p100p010} corresponds to
$(*,*,*,*,*,*,1).$
Thus, $f(111)=1$ implies that $f$ only depends on the first (or the second, or the third) bit and degenerates to a 1-bit Boolean function.

As above, all 3-bit partial Boolean functions with $Q_{E}(f)=1$ have been checked.
\end{proof}

\subsection{Checking all 4-bit partial Boolean functions}\label{fourbits}
In this subsection, we check all 4-bit partial Boolean functions by means of Theorem \ref{acorollaryresult} and obtain the following result.
\begin{corollary}\label{4bitpart1query}
There exist only 10 new non-trivial 4-bit reduced partial Boolean functions with $Q_{E}(f)=1$.
\end{corollary}
\begin{proof}
According to Theorem \ref{acorollaryresult}, a 4-bit partial Boolean function $f$ with $Q_{E}(f)=1$ can be represented by the multilinear polynomial $p(x)=2a_{1}x_{1}+2a_{2}x_{2}+2a_{3}x_{3}+2a_{4}x_{4}$ for some non-negative real numbers $a_{1}, a_{2}, a_{3}, a_{4}$ satisfying $\sum^{4}_{i=1}a_{i}\leq 1$. Without loss of generality, assume that $a_{1}, a_{2}, a_{3}, a_{4}>0$. Thus,
\begin{equation}\label{p1000p0100}
\begin{split}
&(p(1000),\cdots,p(0001),p(1100),\cdots,p(0011),
p(1110),\cdots,p(0111),p(1111))\\
=&(2a_{1},\cdots,2a_{4},2(a_{1}+a_{2}),\cdots,2(a_{3}+a_{4}),2(a_{1}+a_{2}+a_{3}),\cdots,\\
&2(a_{2}+a_{3}+a_{4}),2(a_{1}+a_{2}+a_{3}+a_{4}))\in\{1,*\}^{15}.
\end{split}
\end{equation}
Now, we check all 4-bit multilinear polynomials in Eq.~\eqref{p1000p0100} case by case. For convenience, all cases are listed in Table \ref{Tab04}.

\begin{table}
\begin{center}
\caption{All transformed 4-bit partial Boolean functions with $Q_{E}(f)=1$}\label{Tab04}%
\begin{tabular}{@{}lll@{}}
\toprule
Case & Boolean function: $\{b:f(b)=1\}$  & Remark \\
\midrule
1           &$\{1000,0111\}$                        & $f$ depends on two bits. \\

2           &                          & Similar to Case 1. \\

3: (1.1.1)           &$\{1100,1010,1001,0110,0101,0011\}$                          &  $f$ is symmetric.  \\

{\bf 3: (1.1.2)}           &$\{1100,1010,1001,0111\}$                          &  {\bf Non-trivial}  \\

3: (1.2)           &$\{1100,1010,0110\}$                           &  Included by Case 3: (1.1.1).  \\

3: (1.3)           &$\{1100,1010,0101,0011\}$                          &  Included by Case 3: (1.1.1).  \\

3: (1.4)           &$\{1100,1010,0111\}$                           &  $f$ depends on two bits.  \\

3: (2.1)           &$\{1100,1001,0110,0011\}$                          &  Included by Case 3: (1.1.1).  \\

3: (2.2)           &$\{1100,1001,0101\}$                          &  Included by Case 3: (1.1.1).  \\

3: (2.3)           &$\{1100,1001,0111\}$                          &  Included by Case 3: (1.1.2).  \\

{\bf 3: (3)}           &$\{1100,0110,0101,1011\}$                          &  {\bf Non-trivial} \\

3: (4)           &$\{1100,0101,1011\}$                          & Included by Case 3: (3). \\

3: (5)           &$\{1100,0011\}$                         &  Included by Case 3: (1.1.1).  \\

{\bf 3: (6)}           &$\{1100,1011,0111\}$                         &  {\bf Non-trivial}  \\

4: (1.1)           &$\{1010,1001,0110,0101\}$                         &  Included by Case 3: (1.1.1).  \\

4: (1.2)           &$\{1010,1001,0011\}$                          &  Included by Case 3: (1.1.1).  \\

4: (1.3)           &$\{1010,1001,0111\}$                         &  Included by Case 3: (1.1.2).  \\

{\bf 4: (2)}           &$\{1010,0110,0011,1101\}$                         &  {\bf Non-trivial}  \\

4: (3)           &$\{1010,0101\}$                         &  Included by Case 3: (1.1.1).  \\

4: (4)           &$\{1010,0011,1101\}$                         &  Included by Case 4: (2).  \\

{\bf 4: (5)}           &$\{1010,1101,0111\}$                         &  {\bf Non-trivial}  \\

5: (1)           &$\{1001,0110\}$                          &  Included by Case 3: (1.1.1).  \\

{\bf 5: (2)}           &$\{1001,0101,0011,1110\}$                         &  {\bf Non-trivial}  \\

5: (3)           &$\{1001,0011,1110\}$                         &  Included by Case 5: (2).  \\

{\bf 5: (4)}           &$\{1001,1110,0111\}$                        &  {\bf Non-trivial}  \\

6: (1.1)           &$\{0110,0101,0011\}$                         &  Included by Case 3: (1.1.1).  \\

6: (1.2)           &$\{0110,0101,1011\}$                        &  Included by Case 3: (3).  \\

6: (2)           &$\{0110,0011,1101\}$                         &  Included by Case 4: (2).  \\

{\bf 6: (3)}           &$\{0110,1101,1011\}$                        &  {\bf Non-trivial}  \\

7: (1)           &$\{0101,0011,1110\}$                         &  Included by Case 5: (2).  \\

{\bf 7: (2)}           &$\{0101,1110,1011\}$                        &  {\bf Non-trivial}  \\

{\bf 8:}            &$\{0011,1110,1101\}$                         &  {\bf Non-trivial}  \\

9:            &$\{1110,1101,1011,0111\}$                         &  $f$ is symmetric.  \\

10:            &$\{1111\}$                         &   $f$ depends on one bit. \\
\bottomrule
\end{tabular}
\end{center}
\end{table}

{\bf Case 1:} If $2a_{1}=1$ and $a_{1}, a_{2}, a_{3}, a_{4}>0$ satisfying $\sum^{4}_{i=1}a_{i}\leq 1$, then $f(1x_{2}x_{3}x_{4})=*$ for all $\lvert x_{2}x_{3}x_{4}\rvert\in\{1,2,3\}$ and $f(0x_{2}x_{3}x_{4})=*$ for all $\lvert x_{2}x_{3}x_{4}\rvert\in\{1,2\}$. Thus, the undetermined 4-bit partial Boolean function $f$ represented by Eq.~\eqref{p1000p0100} corresponds to the vector
\begin{equation}\label{}
(1, *, *, *, *,*,*,*,*,*,*,*,*,2(a_{2}+a_{3}+a_{4}),*).
\end{equation}
Set $2(a_{2}+a_{3}+a_{4})=1$ (For example, $a_{1}=\frac{1}{2}$ and $a_{2}=a_{3}=a_{4}=\frac{1}{6}$), then the 4-bit partial Boolean function $f$ corresponds to the vector
\begin{equation}\label{}
(1,*,*,*,*,*,*,*,*,*,*,*,*,1,*)
\end{equation}
and $f(1000)=f(0111)=1$. Obviously, $f$ depends on two bits (i.e., the first and the second bits, or the first and the third bits, or the first and the fourth bits) and can be computed exactly by Deutsch's algorithm.

{\bf Case 2:} The cases that $2a_{1}\neq1,$ $2a_{2}=1$ and $2a_{1},$ $2a_{2}\neq1,$ $2a_{3}=1$ and $2a_{1},$ $2a_{2},$ $2a_{3}\neq1,$ $2a_{4}=1$ are similar to Case 1.

{\bf Case 3:} If $2a_{1},$ $2a_{2},$ $2a_{3},$ $2a_{4}\neq1,$ $2(a_{1}+a_{2})=1$ and $a_{1},$ $a_{2},$ $a_{3},$ $a_{4}>0$ satisfying $\sum^{4}_{i=1}a_{i}\leq 1$, then
the undetermined 4-bit partial Boolean function $f$ represented by Eq.~\eqref{p1000p0100} corresponds to the vector
\begin{equation}\label{a1a3a1a4}
\begin{split}
&(*,*,*,*,1,2(a_{1}+a_{3}),2(a_{1}+a_{4}),2(a_{2}+a_{3}),
2(a_{2}+a_{4}),2(a_{3}+a_{4}),\\
&*,*,2(a_{1}+a_{3}+a_{4}),
2(a_{2}+a_{3}+a_{4}),*).
\end{split}
\end{equation}
Next, we discuss Eq.~\eqref{a1a3a1a4} one by one.
\begin{itemize}
\item [(1)]~If $2(a_{1}+a_{3})=1$, then $a_{2}=a_{3}$. Also, the vector in Eq.~\eqref{a1a3a1a4} degenerates to
\begin{equation}\label{a1a44a2}
\begin{split}
(*,*,*,*,1,1,2(a_{1}{+}a_{4}),4a_{2},2(a_{2}{+}a_{4}),
2(a_{2}{+}a_{4}),*,*,*,2(2a_{2}{+}a_{4}),*).
\end{split}
\end{equation}
In Eq.~\eqref{a1a44a2}, there are four subcases.
   \begin{itemize}
   \item [(1.1)]~If $2(a_{1}+a_{4})=1$, then $a_{2}=a_{3}=a_{4}$. The vector in Eq.~\eqref{a1a44a2} degenerates to
\begin{equation}\label{4a24a24a2}
(*,*,*,*,1,1,1,4a_{2},4a_{2},4a_{2},*,*,*,6a_{2},*).
\end{equation}
In Eq.~\eqref{4a24a24a2}, there are two subcases.
       \begin{itemize}
       \item [(1.1.1)]~Set $4a_{2}=1$ (i.e., $a_{1}=a_{2}=a_{3}=a_{4}=\frac{1}{4}$), then the vector in Eq.~\eqref{4a24a24a2} degenerates to
$(*,*,*,*,1,1,1,1,1,1,*,*,*,*,*).$
       That is, $f(1100)=f(1010)=f(1001)=f(0110)=f(0101)=f(0011)=1$.
       \item [(1.1.2)]~Set $6a_{2}=1$ (i.e., $a_{1}=\frac{1}{3}$ and $a_{2}=a_{3}=a_{4}=\frac{1}{6}$), then the vector in Eq.~\eqref{4a24a24a2} degenerates to
$(*,*,*,*,1,1,1,*,*,*,*,*,*,1,*).$
       That is, $f(1100)=f(1010)=f(1001)=f(0111)=1$.
       \end{itemize}
   \item [(1.2)]~Set $2(a_{1}+a_{4})\neq1$ and $4a_{2}=1$ (i.e., $a_{1}=a_{2}=a_{3}=\frac{1}{4}$ and $0<a_{4}<\frac{1}{4}$), then the vector in Eq.~\eqref{a1a44a2} degenerates to
\begin{equation}\label{}
(*,*,*,*,1,1,*,1,*,*,*,*,*,*,*).
\end{equation}
   That is, $f(1100)=f(1010)=f(0110)=1$.
   \item [(1.3)]~Set $2(a_{1}+a_{4}),$ $4a_{2}\neq1$ and $2(a_{2}+a_{4})=1$ (For example, $a_{1}=a_{4}=\frac{1}{3}$ and $a_{2}=a_{3}=\frac{1}{6}$), then the vector in Eq.~\eqref{a1a44a2} degenerates to
$(*,*,*,*,1,1,*,*,1,1,*,*,*,*,*).$
    That is, $f(1100)=f(1010)=f(0101)=f(0011)=1$. Note that $f$ depends on two bits (i.e., the second and the third bits, or the first and the fourth bits) and can be computed exactly by Deutsch's algorithm.
   \item [(1.4)]~Set $2(a_{1}+a_{4}),$ $4a_{2},$ $2(a_{2}+a_{4})\neq1$ and $2(2a_{2}+a_{4})=1$ (For example, $\frac{1}{3}a_{1}=a_{2}=a_{3}=\frac{1}{2}a_{4}=\frac{1}{8}$), then the vector in Eq.~\eqref{a1a44a2} degenerates to
$(*,*,*,*,1,1,*,*,*,*,*,*,*,1,*).$
    That is, $f(1100)=f(1010)=f(0111)=1$. Note that $f$ depends on two bits (i.e., the first and the fourth bits) and can be computed by Deutsch's algorithm.
   \end{itemize}
\item [(2)]~If $2(a_{1}+a_{3})\neq1$ and $2(a_{1}+a_{4})=1$, then $a_{2}\neq a_{3},$ $a_{3}\neq a_{4}$ and $a_{2}=a_{4}$. And, the vector in Eq.~\eqref{a1a3a1a4} degenerates to
\begin{equation}\label{2a2+a34a2}
\begin{split}
(*,*,*,*,1,*,1,2(a_{2}+a_{3}),4a_{2},2(a_{3}+a_{2}),
*,*,*,2(2a_{2}+a_{3}),*).
\end{split}
\end{equation}
In Eq.~\eqref{2a2+a34a2}, there are three subcases.
   \begin{itemize}
   \item [(2.1)]~Since $2(a_{2}+a_{3})=1$ and $4a_{2}=1$ implies $a_{1}=a_{2}=a_{3}=a_{4}=\frac{1}{4}$ which contradicts with $2(a_{1}+a_{3})\neq1$, $2(a_{2}+a_{3})=1$ and $4a_{2}\neq1$. For example, $a_{1}=a_{3}=\frac{1}{3}$ and $a_{2}=a_{4}=\frac{1}{6}$. And, the vector in Eq.~\eqref{2a2+a34a2} degenerates to
$(*,*,*,*,1,*,1,1,*,1,*,*,*,*,*).$
        Thus, $f(1100)=f(1001)=f(0110)=f(0011)=1$.
   \item [(2.2)]~Set $2(a_{2}+a_{3})\neq1$ and $4a_{2}=1$ (i.e., $a_{1}=a_{2}=a_{4}=\frac{1}{4}$ and $0<a_{3}<\frac{1}{4}$), then the vector in Eq.~\eqref{2a2+a34a2} degenerates to
\begin{equation}\label{}
(*,*,*,*,1,*,1,*,1,*,*,*,*,*,*).
\end{equation}
    Thus, $f(1100)=f(1001)=f(0101)=1$.
   \item [(2.3)]~Set $2(a_{2}+a_{3}),$ $4a_{2}\neq1$ and $2(2a_{2}+a_{3})=1$ (For example, $\frac{1}{3}a_{1}=a_{2}=\frac{1}{2}a_{3}=a_{4}=\frac{1}{8}$), then the vector in Eq.~\eqref{2a2+a34a2} degenerates to
$(*,*,*,*,1,*,1,*,*,*,*,*,*,1,*).$
    Thus, $f(1100)=f(1001)=f(0111)=1.$
   \end{itemize}

\item [(3)]~If $2(a_{1}+a_{3}),$ $2(a_{1}+a_{4})\neq1$ and $2(a_{2}+a_{3})=1$, then $a_{2}\neq a_{3},$ $a_{2}\neq a_{4}$ and $a_{1}=a_{3}$. The vector in Eq.~\eqref{a1a3a1a4} degenerates to
\begin{equation}\label{112a2a42a1a4}
\begin{split}
(*,*,*,*,1,*,*,1,2(a_{2}+a_{4}),*,*,*,
2(2a_{1}+a_{4}),*,*).
\end{split}
\end{equation}
Set $2(a_{2}+a_{4})=1=2(2a_{1}+a_{4})$ (i.e., $a_{1}=a_{3}=a_{4}=\frac{1}{6}$ and $a_{2}=\frac{1}{3}$), then the vector in Eq.~\eqref{112a2a42a1a4} degenerates to
$(*,*,*,*,1,*,*,1,1,*,*,*,1,*,*).$
Thus, $f(1100)=f(0110)=f(0101)=f(1011)=1$.

\item [(4)]~If $2(a_{1}+a_{3}),$ $2(a_{1}+a_{4}),$ $2(a_{2}+a_{3})\neq1$ and $2(a_{2}+a_{4})=1$, then $a_{1}=a_{4}$. The vector in Eq.~\eqref{a1a3a1a4} degenerates to
\begin{equation}\label{1122a1a3}
\begin{split}
(*,*,*,*,1,*,*,*,1,*,*,*,
2(2a_{1}+a_{3}),*,*).
\end{split}
\end{equation}
   Set $2(2a_{1}+a_{3})=1$ (For example, $a_{1}=a_{4}=\frac{7}{24},$ $a_{2}=\frac{5}{24}$ and $a_{3}=\frac{1}{12}$), then the vector in Eq.~\eqref{1122a1a3} degenerates to
$(*,*,*,*,1,*,*,*,1,*,*,*,1,*,*).$
    Thus, $f(1100)=f(0101)=f(1011)=1$.

\item [(5)]~Set $2(a_{1}+a_{3}),$ $2(a_{1}+a_{4}),$ $2(a_{2}+a_{3}),$ $2(a_{2}+a_{4})\neq1$ and $2(a_{3}+a_{4})=1$ (For example, $a_{1}=\frac{1}{6},$ $a_{2}=\frac{1}{3},$ $a_{3}=\frac{1}{8}$ and $a_{4}=\frac{3}{8}$), then the vector in Eq.~\eqref{a1a3a1a4} degenerates to
$(*,*,*,*,1,*,*,*,*,1,*,*,*,*,*).$
 Thus, $f(1100)=f(0011)=1$.

\item [(6)]~Set $2(a_{1}+a_{3}),$ $2(a_{1}+a_{4}),$ $2(a_{2}+a_{3}),$ $2(a_{2}+a_{4}),$ $2(a_{3}+a_{4})\neq1$ and $2(a_{1}+a_{3}+a_{4})=2(a_{2}+a_{3}+a_{4})=1$ (For example, $a_{1}=a_{2}=\frac{1}{4}$ and $a_{3}=a_{4}=\frac{1}{8}$), then the vector in Eq.~\eqref{a1a3a1a4} degenerates to
$(*,*,*,*,1,*,*,*,*,*,*,*,1,1,*).$
Thus, $f(1100)=f(1011)=f(0111)=1$.
\end{itemize}

{\bf Case 4:} If $2a_{1},$ $2a_{2},$ $2a_{3},$ $2a_{4},$ $2(a_{1}+a_{2})\neq1,$ $2(a_{1}+a_{3})=1$ and $a_{1},$ $a_{2},$ $a_{3},$ $a_{4}>0$ satisfying $\sum^{4}_{i=1}a_{i}\leq 1$, then the undetermined 4-bit partial Boolean function $f$ represented by Eq.~\eqref{p1000p0100} corresponds to the vector
\begin{equation}\label{12a1a42a2a3}
\begin{split}
&(*,*,*,*,*,1,2(a_{1}{+}a_{4}),2(a_{2}{+}a_{3}),
2(a_{2}{+}a_{4}),2(a_{3}{+}a_{4}),*,2(a_{1}{+}a_{2}{+}a_{4}),\\
&*,2(a_{2}{+}a_{3}{+}a_{4}),*).
\end{split}
\end{equation}
Next, we discuss Eq.~\eqref{12a1a42a2a3} one by one.
\begin{itemize}
\item [(1)]~If $2(a_{1}+a_{4})=1$, then $a_{3}=a_{4}$. The vector in Eq.~\eqref{12a1a42a2a3} degenerates to
\begin{equation}\label{112a2a32a2a3}
\begin{split}
(*,*,*,*,*,1,1,2(a_{2}+a_{3}),2(a_{2}+a_{3}),
4a_{3},*,*,*,2(a_{2}+2a_{3}),*).
\end{split}
\end{equation}
In Eq.~\eqref{112a2a32a2a3}, there are three subcases.
   \begin{itemize}
   \item [(1.1)]~If $2(a_{2}+a_{3})=1$, then the vector in Eq.~\eqref{112a2a32a2a3} degenerates to
$(*,*,*,*,*,1,1,1,1,4a_{3},*,*,$\\
$*,*,*).$
    Since $4a_{3}=1$ implies $a_{1}=a_{2}=a_{3}=a_{4}=$ $\frac{1}{4}$ which contradicts with the condition $2(a_{1}+a_{2})\neq1$, $4a_{3}\neq1$ and $f(1010)=f(1001)=f(0110)=f(0101)=1$. In this case, we can set $a_{1}=a_{2}=\frac{3}{8}$ and $a_{3}=a_{4}=\frac{1}{8}$.
   \item [(1.2)]~Set $2(a_{2}+a_{3})\neq1$ and $4a_{3}=1$ (For example, $a_{1}=a_{3}=a_{4}=\frac{1}{4}$ and $a_{2}=\frac{1}{8}$ ), then the vector in Eq.~\eqref{112a2a32a2a3} degenerates to
\begin{equation}\label{}
(*,*,*,*,*,1,1,*,*,1,*,*,*,*,*).
\end{equation}
    Thus, $f(1010)=f(1001)=f(0011)=1$.
   \item [(1.3)]~Set $2(a_{2}+a_{3}),$ $4a_{3}\neq1$ and $2(a_{2}+2a_{3})=1$ (For example, $a_{1}=\frac{5}{12},$ $a_{2}=\frac{1}{3},$ $a_{3}=a_{4}=\frac{1}{12}$), then the vector in Eq.~\eqref{112a2a32a2a3} degenerates to
$(*,*,*,*,*,1,1,*,*,*,*,*,*,1,*).$
Thus, $f(1010)=f(1001)=f(0111)=1$.
   \end{itemize}
\item [(2)]~If $2(a_{1}+a_{4})\neq1$ and $2(a_{2}+a_{3})=1$, then $a_{1}=a_{2}\neq a_{3}$ and $a_{3}\neq a_{4}$. Thus, the vector in Eq.~\eqref{12a1a42a2a3} degenerates to
\begin{equation}\label{112a3a422a1a4}
\begin{split}
(*,*,*,*,*,1,*,1,*,2(a_{3}+a_{4}),*,
2(2a_{1}+a_{4}),*,*,*).
\end{split}
\end{equation}
Set $2(a_{3}+a_{4})=1=2(2a_{1}+a_{4})$ (For example, $a_{1}=a_{2}=\frac{1}{2}a_{3}=a_{4}=\frac{1}{6}$), then the vector in Eq.~\eqref{112a3a422a1a4} degenerates to
\begin{equation}\label{}
(*,*,*,*,*,1,*,1,*,1,*,1,*,*,*)
\end{equation}
and $f(1010)=f(0110)=f(0011)=f(1101)=1$.
\item [(3)]~If $2(a_{1}+a_{4}),$ $2(a_{2}+a_{3})\neq1$ and $2(a_{2}+a_{4})=1$, then $a_{3}\neq a_{2}$ implies that the vector in Eq.~\eqref{12a1a42a2a3} degenerates to
$(*,*,*,*,*,1,*,*,1,*,*,*,*,*,*).$
Thus, $f(1010)=f(0101)=1$. In this case, we can set $a_{1}=\frac{1}{8},$ $a_{2}=\frac{1}{6},$ $a_{3}=\frac{3}{8}$ and $a_{4}=\frac{1}{3}$.
\item [(4)]~Set $2(a_{1}+a_{4}),$ $2(a_{2}+a_{3}),$ $2(a_{2}+a_{4})\neq1$ and $2(a_{3}+a_{4})=1=2(a_{1}+a_{2}+a_{4})$ (For example, $a_{1}=\frac{1}{12},$ $a_{2}=\frac{1}{3},$ $a_{3}=\frac{5}{12},$ $a_{4}=\frac{1}{12}$), then the vector in Eq.~\eqref{12a1a42a2a3} degenerates to
$(*,*,*,*,*,1,*,*,*,1,*,1,*,*,*)$
and $f(1010)=f(0011)=f(1101)=1$.
\item [(5)]~Set $2(a_{1}+a_{4}),$ $2(a_{2}+a_{3}),$ $2(a_{2}+a_{4}),$ $2(a_{3}+a_{4})\neq1$ and $2(a_{1}+a_{2}+a_{4})=2(a_{2}+a_{3}+a_{4})=1$ (i.e., $a_{1}=a_{3}=a_{2}+a_{4}=\frac{1}{4}$), then the vector in Eq.~\eqref{12a1a42a2a3} degenerates to
$(*,*,*,*,*,1,*,*,*,*,*,1,*,1,*).$
Thus, $f(1010)=f(1101)=f(0111)=1$.
\end{itemize}

{\bf Case 5:} If $2a_{1},$ $2a_{2},$ $2a_{3},$ $2a_{4},$ $2(a_{1}+a_{2}),$ $2(a_{1}+a_{3})\neq1,$ $2(a_{1}+a_{4})=1$ and $a_{1},$ $a_{2},$ $a_{3},$ $a_{4}>0$ satisfying $\sum^{4}_{i=1}a_{i}\leq 1$, then
the undetermined 4-bit partial Boolean function $f$ represented by Eq.~\eqref{p1000p0100} corresponds to the vector
\begin{equation}\label{12a2a32a2a42a3a4111}
\begin{split}
(*,*,*,*,*,*,1,2(a_{2}{+}a_{3}),2(a_{2}{+}a_{4}),
2(a_{3}{+}a_{4}),2(a_{1}{+}a_{2}{+}a_{3}),*,*,
2(a_{2}{+}a_{3}{+}a_{4}),*).
\end{split}
\end{equation}
Next, we discuss Eq.~\eqref{12a2a32a2a42a3a4111} one by one.
\begin{itemize}
\item [(1)]~If $2(a_{2}+a_{3})=1$, then the vector in Eq.~\eqref{12a2a32a2a42a3a4111} degenerates to
\begin{equation}\label{}
(*,*,*,*,*,*,1,1,2(a_{2}+a_{4}),2(a_{3}+a_{4}),*,*,*,*,*).
\end{equation}
And, $a_{1}\neq a_{2}$ (implied by $2(a_{1}+a_{3})\neq1$ and $2(a_{2}+a_{3})=1$) and $2(a_{1}+a_{4})=1$ imply that $2(a_{2}+a_{4})\neq 1$.
Next, $a_{1}\neq a_{3}$ (implied by $2(a_{1}+a_{2})\neq1$ and $2(a_{2}+a_{3})=1$) and $2(a_{1}+a_{4})=1$ imply that $2(a_{3}+a_{4})\neq 1$.
Thus, $f(1001)=f(0110)=1$. In this case, we can set $a_{1}=\frac{1}{8},$ $a_{2}=\frac{1}{6},$ $a_{3}=\frac{1}{3},$ $a_{4}=\frac{3}{8}$.
\item [(2)]~If $2(a_{2}+a_{3})\neq1$ and $2(a_{2}+a_{4})=1$, then the vector in Eq.~\eqref{12a2a32a2a42a3a4111} degenerates to
\begin{equation}\label{112a3a42a1a2a3}
\begin{split}
(*,*,*,*,*,*,1,*,1,2(a_{3}+a_{4}),
2(a_{1}+a_{2}+a_{3}),*,*,*,*).
\end{split}
\end{equation}
Set $2(a_{3}+a_{4})=2(a_{1}+a_{2}+a_{3})=1$ (For example, $a_{1}=a_{2}=a_{3}=\frac{1}{6}$ and $a_{4}=\frac{1}{3}$), then the vector in Eq.~\eqref{112a3a42a1a2a3} degenerates to
\begin{equation}\label{}
(*,*,*,*,*,*,1,*,1,1,1,*,*,*,*)
\end{equation}
and $f(1001)=f(0101)=f(0011)=f(1110)=1$.
\item [(3)]~If $2(a_{2}+a_{3}),$ $2(a_{2}+a_{4})\neq1$ and $2(a_{3}+a_{4})=1$, then the vector in Eq.~\eqref{12a2a32a2a42a3a4111} degenerates to
\begin{equation}\label{112a1a2a3***}
\begin{split}
(*,*,*,*,*,*,1,*,*,1,2(a_{1}+a_{2}+a_{3}),
*,*,*,*).
\end{split}
\end{equation}
Set $2(a_{1}+a_{2}+a_{3})=1$ (For example, $a_{1}=a_{3}=\frac{1}{12},$ $a_{2}=\frac{1}{3},$ $a_{4}=\frac{5}{12}$), then the vector in Eq.~\eqref{112a1a2a3***} degenerates to
$(*,*,*,*,*,*,1,*,*,1,1,*,*,*,*)$
and $f(1001)=f(0011)=f(1110)=1$.
\item [(4)]~Set $2(a_{2}+a_{3}),$ $2(a_{2}+a_{4}),$ $2(a_{3}+a_{4})\neq1$ and $2(a_{1}+a_{2}+a_{3})=2(a_{2}+a_{3}+a_{4})=1$ (For example, $a_{1}=a_{4}=\frac{1}{4}=a_{2}+a_{3}$), then the vector in Eq.~\eqref{12a2a32a2a42a3a4111} degenerates to
$(*,*,*,*,*,*,1,*,*,*,1,*,*,1,*)$
and $f(1001)=f(1110)=f(0111)=1$.
\end{itemize}

{\bf Case 6:} If $2a_{1},$ $2a_{2},$ $2a_{3},$ $2a_{4},$ $2(a_{1}+a_{2}),$ $2(a_{1}+a_{3}),$ $2(a_{1}+a_{4})\neq1,$ $2(a_{2}+a_{3})=1$ and $a_{1},$ $a_{2},$ $a_{3},$ $a_{4}>0$ satisfying $\sum^{4}_{i=1}a_{i}\leq 1$, then
the undetermined 4-bit partial Boolean function $f$ represented by Eq.~\eqref{p1000p0100} corresponds to the vector
\begin{equation}\label{case612a2a42a3a4}
\begin{split}
(*,*,*,*,*,*,*,1,2(a_{2}{+}a_{4}),2(a_{3}{+}a_{4}),*,
2(a_{1}{+}a_{2}{+}a_{4}),2(a_{1}{+}a_{3}{+}a_{4}),*,*).
\end{split}
\end{equation}
Next, we discuss Eq.~\eqref{case612a2a42a3a4} one by one.
\begin{itemize}
\item [(1)]~If $2(a_{2}+a_{4})=1$, then $a_{3}=a_{4}$ and the vector in Eq.~\eqref{case612a2a42a3a4} degenerates to
\begin{equation}\label{***114a32a12a3}
(*,*,*,*,*,*,*,1,1,4a_{3},*,*,2(a_{1}+2a_{3}),*,*).
\end{equation}
In Eq.~\eqref{***114a32a12a3}, there are two subcases.
\begin{itemize}
\item [(1.1)]~Set $4a_{3}=1$ (i.e., $a_{1}<a_{2}=a_{3}=a_{4}=\frac{1}{4}$), then the vector in Eq.~\eqref{***114a32a12a3} degenerates to
$(*,*,*,*,*,*,*,1,1,1,*,*,*,*,*)$
and $f(0110)=f(0101)=f(0011)=1$.
\item [(1.2)]~Set $2(a_{1}+2a_{3})=1$ (For example, $a_{1}=\frac{1}{12},$ $a_{2}=\frac{7}{24},$ $a_{3}=a_{4}=\frac{5}{24}$), then the vector in Eq.~\eqref{***114a32a12a3} degenerates to
\begin{equation}\label{}
(*,*,*,*,*,*,*,1,1,*,*,*,1,*,*)
\end{equation}
and $f(0110)=f(0101)=f(1011)=1$.
\end{itemize}
\item [(2)]~If $2(a_{2}+a_{4})\neq1$ and $2(a_{3}+a_{4})=1$, then $a_{2}=a_{4}$ and the vector in Eq.~\eqref{case612a2a42a3a4} degenerates to
\begin{equation}\label{*******112a1a2a4}
(*,*,*,*,*,*,*,1,*,1,*,2(a_{1}+a_{2}+a_{4}),*,*,*).
\end{equation}
Set $2(a_{1}+a_{2}+a_{4})=1$ (For example, $a_{1}=\frac{1}{12},$ $a_{2}=\frac{5}{24}=a_{4},$ $a_{3}=\frac{7}{24}$), then the vector in Eq.~\eqref{*******112a1a2a4} degenerates to
$(*,*,*,*,*,*,*,1,*,1,*,1,*,*,*)$
and $f(0110)=f(0011)=f(1101)=1$.
\item [(3)]~If $2(a_{2}+a_{4}),$ $2(a_{3}+a_{4})\neq1$ and $2(a_{1}+a_{2}+a_{4})=1$, then the vector in Eq.~\eqref{case612a2a42a3a4} degenerates to
\begin{equation}\label{****112a1a3a4}
(*,*,*,*,*,*,*,1,*,*,*,1,2(a_{1}+a_{3}+a_{4}),*,*).
\end{equation}
Set $2(a_{1}+a_{3}+a_{4})=1$ (i.e., $a_{2}=a_{3}=a_{1}+a_{4}=\frac{1}{4}$), then the vector in Eq.~\eqref{****112a1a3a4} degenerates to
$(*,*,*,*,*,*,*,1,*,*,*,1,1,*,*)$
and $f(0110)=f(1101)=f(1011)=1$.
\end{itemize}

{\bf Case 7:} If $2a_{1},$ $2a_{2},$ $2a_{3},$ $2a_{4},$ $2(a_{1}+a_{2}),$ $2(a_{1}+a_{3}),$ $2(a_{1}+a_{4}),$ $2(a_{2}+a_{3})\neq1,$ $2(a_{2}+a_{4})=1$ and $a_{1},$ $a_{2},$ $a_{3},$ $a_{4}>0$ satisfying $\sum^{4}_{i=1}a_{i}\leq 1$, then
the undetermined 4-bit partial Boolean function $f$ represented by Eq.~\eqref{p1000p0100} corresponds to the vector
\begin{equation}\label{******12a3a42a1a2a3*}
\begin{split}
(*,*,*,*,*,*,*,*,1,2(a_{3}+a_{4}),2(a_{1}+a_{2}+a_{3}),
*,2(a_{1}+a_{3}+a_{4}),*,*).
\end{split}
\end{equation}
Next, we discuss Eq.~\eqref{******12a3a42a1a2a3*} one by one.
\begin{itemize}
\item [(1)]~If $2(a_{3}+a_{4})=1$, then $a_{2}=a_{3}$ and the vector in Eq.~\eqref{******12a3a42a1a2a3*} degenerates to
\begin{equation}\label{******112a1a2a3***}
(*,*,*,*,*,*,*,*,1,1,2(a_{1}+a_{2}+a_{3}),*,*,*,*).
\end{equation}
Set $2(a_{1}+a_{2}+a_{3})=1$ (For example, $a_{1}=\frac{1}{12},$ $a_{2}=a_{3}=\frac{5}{24},$ $a_{4}=\frac{7}{24}$), then the vector in Eq.~\eqref{******112a1a2a3***} degenerates to
$(*,*,*,*,*,*,*,*,1,1,1,*,*,*,*)$
and $f(0101)=f(0011)=f(1110)=1$.
\item [(2)]~If $2(a_{3}+a_{4})\neq1$ and $2(a_{1}+a_{2}+a_{3})=1$, then $a_{4}=a_{1}+a_{3}$ and the vector in Eq.~\eqref{******12a3a42a1a2a3*} degenerates to
\begin{equation}\label{*****1*1*2a1a3a4}
(*,*,*,*,*,*,*,*,1,*,1,*,2(a_{1}+a_{3}+a_{4}),*,*).
\end{equation}
Set $2(a_{1}+a_{3}+a_{4})=1$ (i.e., $a_{2}=a_{4}=\frac{1}{4}=a_{1}+a_{3}$), then the vector in Eq.~\eqref{*****1*1*2a1a3a4} degenerates to
$(*,*,*,*,*,*,*,*,1,*,1,*,1,*,*)$
and $f(0101)=f(1110)=f(1011)=1$.
\end{itemize}

{\bf Case 8:} If $2a_{1},$ $2a_{2},$ $2a_{3},$ $2a_{4},$ $2(a_{1}+a_{2}),$ $2(a_{1}+a_{3}),$ $2(a_{1}+a_{4}),$ $2(a_{2}+a_{3}),$ $2(a_{2}+a_{4})\neq1,$ $2(a_{3}+a_{4})=1$ and $a_{1},$ $a_{2},$ $a_{3},$ $a_{4}>0$ satisfying $\sum^{4}_{i=1}a_{i}\leq 1$, then the undetermined 4-bit partial Boolean function $f$ represented by Eq.~\eqref{p1000p0100} corresponds to the vector
\begin{equation}\label{******12a1a2a32a1a2a4***}
\begin{split}
(*,*,*,*,*,*,*,*,*,1,2(a_{1}+a_{2}+a_{3}),
2(a_{1}+a_{2}+a_{4}),*,*,*).
\end{split}
\end{equation}
Set $2(a_{1}+a_{2}+a_{3})=2(a_{1}+a_{2}+a_{4})=1$ (i.e., $a_{3}=a_{4}=a_{1}+a_{2}=\frac{1}{4}$), then the vector in Eq.~\eqref{******12a1a2a32a1a2a4***} degenerates to
$(*,*,*,*,*,*,*,*,*,1,1,1,*,*,*)$
and $f(0011)=f(1110)=f(1101)=1$.

{\bf Case 9:} If $2a_{1},$ $2a_{2},$ $2a_{3},$ $2a_{4},$ $2(a_{1}+a_{2}),$ $2(a_{1}+a_{3}),$ $2(a_{1}+a_{4}),$ $2(a_{2}+a_{3}),$ $2(a_{2}+a_{4}),$ $2(a_{3}+a_{4})\neq1,$ $2(a_{1}+a_{2}+a_{3})=1$ and $a_{1},$ $a_{2},$ $a_{3}$ $a_{4}>0$ satisfying $\sum^{4}_{i=1}a_{i}\leq 1$, then the undetermined 4-bit partial Boolean function $f$ represented by Eq.~\eqref{p1000p0100} corresponds to the vector
\begin{equation}\label{************12a1a2a42a1a3a42}
\begin{split}
(*,*,*,*,*,*,*,*,*,*,1,2(a_{1}{+}a_{2}{+}a_{4}),
2(a_{1}{+}a_{3}{+}a_{4}),2(a_{2}{+}a_{3}{+}a_{4}),*).
\end{split}
\end{equation}
Set $2(a_{1}+a_{2}+a_{4})=2(a_{1}+a_{3}+a_{4})=2(a_{2}+a_{3}+a_{4})=1$ (i.e., $a_{1}=a_{2}=a_{3}=a_{4}=\frac{1}{6}$), then the vector in Eq.~\eqref{************12a1a2a42a1a3a42} degenerates to
\begin{equation}\label{}
(*,*,*,*,*,*,*,*,*,*,1,1,1,1,*)
\end{equation}
and $f(1110)=f(1101)=f(1011)=f(0111)=1$.

{\bf Case 10:} If $2a_{1},$ $2a_{2},$ $2a_{3},$ $2a_{4},$ $2(a_{1}+a_{2}),$ $2(a_{1}+a_{3}),$ $2(a_{1}+a_{4}),$ $2(a_{2}+a_{3}),$ $2(a_{2}+a_{4}),$ $2(a_{3}+a_{4}),$ $2(a_{1}+a_{2}+a_{3}),$ $2(a_{1}+a_{2}+a_{4}),$ $2(a_{1}+a_{3}+a_{4}),$ $2(a_{2}+a_{3}+a_{4})\neq1,$ $2(a_{1}+a_{2}+a_{3}+a_{4})=1$ (For example, $a_{1}=a_{2}=a_{3}=a_{4}=\frac{1}{8}$), then the undetermined 4-bit partial Boolean function $f$ represented by Eq.~\eqref{p1000p0100} corresponds to the vector
$(*,*,*,*,*,*,*,*,*,*,*,*,*,*,1)$
and $f(1111)=1$.

As above, all 4-bit partial Boolean functions computed by exact quantum 1-query algorithms have been checked and by virtue of Table \ref{Tab04} the proof is completed.
\end{proof}

\section{Applications of Theorem \ref{Constructiontheorem}}\label{Applicationoftheconstructionmethod}
In this section, we present two applications of the Construction theorem (i.e., Theorem \ref{Constructiontheorem}).
\subsection{Finding new non-trivial examples}\label{somenewexamples}

In 2021, Ref. \cite{Qiuxu2021SPIN} presents a partial Boolean function with $Q_{E}(f)=1$ which cannot be computed exactly by Deutsch-Jozsa algorithm. In this subsection  we point out that this task can be done by Theorem \ref{Constructiontheorem}, as well.

Now, we present an example (different from the example of Ref. \cite{Qiuxu2021SPIN}) in the following.
In Definition \ref{def000}, take $n=3k$, $z_{1}=\cdots=z_{k}=\frac{1}{n}$ and $z_{k+1}=\cdots=z_{n}=\frac{1}{2n}$. According to Theorem \ref{Constructiontheorem}, we have
\begin{equation}\label{}
\sum_{i:x_{i}=1}z_{i}{=}\frac{1}{n}\lvert x_{1}\cdots x_{k}\rvert{+}\frac{1}{2n}\lvert x_{k+1}\cdots x_{3k}\rvert{=}\frac{1}{2}.
\end{equation}
If an undetermined $n$-bit partial Boolean function $f: D\rightarrow\{0,1\}$ can be computed exactly by the quantum 1-query algorithm, then the equation
\begin{equation}\label{theequality1111111}
2\lvert x_{1}\cdots x_{k}\rvert+\lvert x_{k+1}\cdots x_{3k}\rvert=3k
\end{equation}
holds for all $x\in D\backslash\{00\cdots 0\}$.
\begin{itemize}
\item [(1)] For $k=1$, the 3-bit partial Boolean function
\begin{equation}\label{example1111111111}
f(x)=\left\{
\begin{array}{lcl}
1, & & \lvert x_{1}\rvert=\lvert x_{2}x_{3}\rvert=1,\\
0, & & x=00...0,\\
*, & & \textrm{else}
\end{array} \right.
\end{equation}
depends on the first bit and is trivial.
\item [(2)] For $k=2$, $S$ becomes $\{x:\lvert x_{1}x_{2}\rvert=1~\textrm{and}~\lvert x_{3}x_{4}x_{5}x_{6}\rvert=4,~\textrm{or}~\lvert x_{1}x_{2}\rvert=2~\textrm{and}~\lvert x_{3}x_{4}x_{5}x_{6}\rvert=2\}$.
Then, the 6-bit partial Boolean function
\begin{equation}\label{example2222222222}
f(x)=\left\{
\begin{array}{lcl}
1, & & x\in S,\\
0, & & x=00...0,\\
*, & & \textrm{else}
\end{array} \right.
\end{equation}
is non-trivial.
\item [(3)] For $k=3$, $S$ becomes $\{x:\lvert x_{1}x_{2}x_{3}\rvert=2~\textrm{and}~\lvert x_{4}x_{5}\cdots x_{9}\rvert=5,~\textrm{or}~\lvert x_{1}x_{2}x_{3}\rvert=3~\textrm{and}~\lvert x_{4}x_{5}\cdots x_{9}\rvert=3\}.$ Then, the 9-bit partial Boolean function
\begin{equation}\label{example3333333333}
f(x)=\left\{
\begin{array}{lcl}
1, & &x\in S,\\
0, & & x=00...0,\\
*, & & \textrm{else}
\end{array} \right.
\end{equation}
is non-trivial.
\end{itemize}
As above, Theorem \ref{Constructiontheorem} contributes an efficient method for discovering more new partial Boolean functions with $Q_{E}(f)=1$.

\subsection{The computational power of Deutsch-Jozsa algorithm}\label{DJcomputesnoo}

In 2016, Qiu et al. \cite{Qiu2016revisit1,Qiu2016character} proved that any symmetric partial Boolean function $f$ has exact quantum 1-query complexity if and only if $f$ can be computed exactly by Deutsch-Jozsa algorithm. In other word, Deutsch-Jozsa algorithm is available for all symmetric partial Boolean functions with $Q_{E}(f)=1$. Thus, an interesting and natural question is what other partial Boolean functions can be computed exactly by Deutsch-Jozsa algorithm? In this subsection, we answer this problem in terms of the following theorem.
\begin{theorem}\label{djpower}
All partial Boolean functions computed by Deutsch-Jozsa algorithm can be transformed (up to the equivalence in Definition \ref{def111000}) to a symmetric partial Boolean function with $Q_{E}(f)=1$.
\end{theorem}
\begin{proof}
First, Deutsch-Jozsa algorithm corresponds to the case that $p=1$ and $an\in[\frac{1}{2},1]$ in Definition \ref{def000} and Theorem \ref{Constructiontheorem}, since  Deutsch-Jozsa algorithm produces an equal superposition of all computational basis states, before its only query operator is performed. Then, the set
\begin{equation}\label{allpart000}
T=\{x:a\lvert x\rvert=\frac{1}{2}\}=\{x:\lvert x\rvert=\frac{1}{2a}\}.
\end{equation}
By Theorem \ref{Constructiontheorem}, partial Boolean functions (up to the equivalence in Definition \ref{def111000}) computed exactly by Deutsch-Jozsa algorithm are in the form
\begin{equation}\label{functiondeutschJ}
f(x)=\left\{
\begin{array}{lcl}
1~\textrm{or}~*, & & \lvert x\rvert=\frac{1}{2a},\\
0, & & x=00...0,\\
*, & & \textrm{else}
\end{array} \right.
\end{equation}
where the Hamming weight $\frac{1}{2a}$ is an integer in the set $\in\{1,2,3,\cdots,n\}$. Note that $an\in[\frac{1}{2},1]$, and thus $\frac{1}{2a}\in[\frac{n}{2},n]$. In fact, these functions in Eq.~\eqref{functiondeutschJ} are symmetric partial Boolean functions with $Q_{E}(f)=1$ and had been founded out by Qiu et al. \cite{Qiu2016revisit1,Qiu2016character}.
\end{proof}


\end{document}